\documentclass[journal,10pt]{IEEEtran}
\usepackage{graphicx}
\usepackage{epsf}
\usepackage{epsfig}
\usepackage{epstopdf}
\usepackage{ifpdf}
\usepackage{cite}
\usepackage[utf8]{luainputenc}
\usepackage{amssymb}
\usepackage{amsthm}
\usepackage{lipsum}

\newtheorem{definition}{\bf Definition}

\newtheorem{proposition}{\bf Proposition}

\usepackage{pifont}
\newcommand{\cmark}{\text{\ding{51}}}
\newcommand{\xmark}{\text{\ding{55}}}

\ifCLASSINFOpdf
\else
\fi

\usepackage[cmex10]{amsmath}
\usepackage{algorithmic}
\usepackage{array}
\usepackage{mdwmath}
\usepackage{mdwtab}
\usepackage{eqparbox}
\usepackage{fixltx2e}
\usepackage{stfloats}
\usepackage{url}

\usepackage{breakurl}
\usepackage[breaklinks]{hyperref}
\usepackage{caption}
\usepackage{subcaption}
\usepackage{amsfonts}
\usepackage{bbm}
\usepackage{breakurl}
\usepackage{doi}
\usepackage{float}
\usepackage{epsf}
\usepackage{epsfig}
\usepackage{epstopdf}
\usepackage{algorithmic}
\usepackage[linesnumbered, ruled, vlined]{algorithm2e}
\SetKw{KwOr}{or}
\SetKw{KwAnd}{and}
\SetKw{Cnt}{continue}

\begin{document}
\title{Heterogeneous Drone Small Cells: \\ Optimal 3D Placement for Downlink Power Efficiency and Rate Satisfaction}

\author{Nima Namvar,~\IEEEmembership{member,~IEEE}, Fatemeh Afghah~\IEEEmembership{Senior member,~IEEE}, Ismail Guvenc, ~\IEEEmembership{Fellow,~IEEE}

\thanks{N. Namvar is with the School of Informatics, Computing and Cyber Systems (SICSS), Northern Arizona University, Flagstaff, AZ, USA e-mail: nima.namvar@nau.edu, F. Afghah is with the Department of Electrical \& Computer Engineering, Clemson University e-mail: fafghah@clemson.edu, 
I. Guvenc is with the Dept. Electrical and Computer Engineering, North Carolina State University, Raleigh, NC, USA e-mail: iguvenc@ncsu.edu.}
\thanks{Manuscript received XXX, XX, 2015; revised XXX, XX, 2015.}}


\markboth{IEEE Transactions on Vehicular Technology,~Vol.~XX, No.~XX, XXX~2022}
{}

\maketitle

\begin{abstract}
In this paper, we consider a heterogeneous repository of drone-enabled aerial base stations with varying transmit powers that provide downlink wireless coverage for ground users. One particular challenge is optimal selection and deployment of a subset of available drone base stations (DBSs) to satisfy the downlink data rate requirements while minimizing the overall power consumption. In order to address this challenge, we formulate an optimization problem to select the best subset of available DBSs so as to guarantee wireless coverage with some acceptable transmission rate in the downlink path. In addition to the selection of DBSs, we determine their 3D position so as to minimize their overall power consumption. Moreover, assuming that the DBSs operate in the same frequency band, we develop a novel and computationally-efficient beamforming method to alleviate the inter-cell interference impact on the downlink. We propose a Kalai–Smorodinsky bargaining solution to determine the optimal beamforming strategy in the downlink path to compensate for the impairment caused by the interference. Simulation results demonstrate the effectiveness of the proposed solution and provide valuable insights into the performance of the heterogeneous drone-based small cell networks.
\end{abstract}


\IEEEpeerreviewmaketitle

\section{Introduction}
The recent advances in the area of unmanned aerial vehicles (UAVs), commonly known as drones, have made it possible to widely deploy drones across a wide variety of  application domains ranging from surveillance to shipping and delivery, disaster management, geographic mapping, search and rescue, and wireless networking \cite{saad2020wireless}. In particular, the cellular telecommunications may avail from drone-mounted aerial base stations to satisfy the coverage and rate requirements of wireless users in areas which lack coverage or are heavily congested, such as hotspot areas \cite{azeri2019tutorial, mozaffari2019tutorial}. 

The altitude dimension and mobility which stem from the flying nature of drone base stations (DBSs), provide new degrees of freedom that a network operator can exploit to improve the design of airborne cellular systems. For instance, compared to terrestrial base stations, the DBSs benefit from a much higher likelihood of establishing line-of-sight (LoS) links towards ground stations by adjusting their altitude \cite{orsino2017effects}. Furthermore, DBSs are better equipped to cope with the mobility of ground users and environmental changes as compared to fixed ground base stations. All things considered, the salient attributes of DBSs such as flexible and on-demand deployment, strong LoS connection links, and additional design degrees of freedom, make them a promising solution to facilitate the vision of ubiquitous connectivity and enhanced network capacity in the next generation of broadband cellular networks \cite{li2018uav}. For example, Qualcomm has already announced its plan to employ DBSs as an enabler for everywhere anytime wireless connectivity in the upcoming fifth generation (5G) wireless networks \cite{Qualcomm}. Meanwhile, the "Flying Cow" project by AT\&T \cite{Aquila} leverages the UAV technology to create an aerial wireless network which provides ubiquitous Internet access to rural and remote areas up to 4G-LTE speeds.

Developing fully fledged drone-based wireless networks brings forward unique technical challenges that are rooted in unique features of DBSs which are inherently different from those of the conventional ground base stations:
\begin{itemize}
    \item The air-to-ground (AtG) wireless channel poses a new propagation environment with different characteristics compared to those of terrestrial channels. For instance, the high likelihood of existing a strong LoS component in the received signal as well as the airframe  shadowing  caused  by  the  structural design and rotation of the DBSs, are some unique features of the AtG channel \cite{SurveyChannel}.
    \item The deployment of DBSs is naturally done in 3D space and the flying nature of DBSs  mandates dynamic optimization of their 3D location for improved performance.
    \item When designing a drone-based wireless network, it is imperative to take into account the different power capabilities of the DBSs. Indeed, depending on the type of drone being utilized as a DBS, the range of transmit power and battery capacity can differ substantially from one DBS to another. These capabilities can directly impact the quality of service (QoS) that the DBSs provide for ground users \cite{galkin2019uavs}.
    \item Managing the network interference becomes more challenging when the DBSs are deployed to provide connectivity for ground users. This is partially due to the fact that ground users may receive strong LoS signals from multiple DBSs, which can significantly degrade the quality of the intended signal \cite{challita2019interference}. On the other hand, the limited on-board energy of the DBSs calls for effective and yet, computationally efficient algorithms to address the interference \cite{fouda2019interference}.
\end{itemize}

In this paper, a novel technique is developed for optimally selecting and deploying a heterogeneous set of DBSs to provide wireless coverage for ground users while minimizing the aggregate DBSs' transmit power needed to satisfy the downlink data rate. In particular, our contributions in this work can be
summarized as follows:
\begin{itemize}
    \item Considering a repository of heterogeneous DBSs with varying transmit power and flight altitude, we jointly derive the optimal resource allocation strategy (i.e., selecting a subset of available DBSs) and 3D placement of the DBSs. The goal is to minimize the total transmit power while maintaining the desired downlink data rate. In contrast to the existing literature, the type and the number of DBSs that need to be deployed is not known a priori.
    \item  Assuming that the DBSs are operating in the same spectrum, we devise a novel beamforming method based on the Nash bargaining game to alleviate the impact of inter-cell interference between the DBSs.
\end{itemize}

To this end, we decompose the optimization problem into two subproblems that will be solved iteratively. In the first problem, assuming that the DBSs are equipped with a directional antenna, we solve an optimization problem to determine the best subset of DBSs as well as their corresponding 3D location that can provide a reasonable signal to noise ratio (SNR) at the ground receivers. In the second subproblem, given the topology of the network resulting from the first subproblem, we propose a bargaining game between the interfering DBSs to find the optimal downlink beamforming that increases the data rate in the interference channel. If the optimal beamforming resulted from the second subproblem achieves the required data rate threshold, the iteration is stopped. Otherwise, the topology of the DBSs needs to be adjusted to reduce the interference. This is an iterative process in which the results of each subproblem are used in the other subproblem for the next iteration. These computations are performed by the control center until the location of the DBSs, device association, and transmit power of the DBSs are obtained.

\begin{table*}[]
\caption{Literature review for DBS placement optimization.}
\label{tab:literature}
\setlength{\tabcolsep}{3pt}
\begin{tabular}{llllllll}
\hline
\textbf{Ref.}      & \multicolumn{1}{l}{\textbf{Goal}}                                                                                                                        & \begin{tabular}[l]{@{}c@{}}\textbf{Downlink rate} \\ \textbf{satisfaction}\end{tabular} & \multicolumn{1}{l}{\begin{tabular}[l]{@{}l@{}}\textbf{DBS power}\\ \textbf{optimization}\end{tabular}} & \begin{tabular}[l]{@{}l@{}}\textbf{Interference} \\ \textbf{management}\end{tabular} & \begin{tabular}[l]{@{}l@{}}\textbf{Heterogeneous} \\ \textbf{DBSs}\end{tabular} & \begin{tabular}[l]{@{}l@{}}\textbf{Predetermined}\\ \textbf{number of DBSs}\end{tabular} & \begin{tabular}[l]{@{}l@{}}\textbf{3D} \\ \textbf{placement}\end{tabular} \\ \hline
\cite{Azari}        & \begin{tabular}[c]{@{}l@{}}Optimal flight altitude of a DBS\\ for a target outage probability\\ in a Rician fading channel\end{tabular}           & $\xmark$                                                                     & $\xmark$                                                                              & $\xmark$                                                           & $\xmark$                                                            & $\cmark$                                                                   & $\xmark$                                                      \\ \hline
\cite{al2014optimal}        & \begin{tabular}[c]{@{}l@{}}Optimal altitude of a single DBS\\ for maximum coverage area\end{tabular}                                            & $\xmark$                                                                     & $\xmark$                                                                              & $\xmark$                                                           & $\xmark$                                                            & $\cmark$                                                                    & $\xmark$                                                      \\ \hline
\cite{mozaffariDesign}        & \begin{tabular}[c]{@{}l@{}}Optimal distance between two\\ interfering DBSs for best coverage\\ performance\end{tabular}                           & $\xmark$                                                                     & $\cmark$                                                                                   & $\xmark$                                                      & $\xmark$                                                            & $\cmark$                                                                   & $\cmark$                                                     \\ \hline
\cite{mozaffariLetter}        & \begin{tabular}[c]{@{}l@{}}Optimal placement of symmetric\\ DBSs with same altitude and transmit\\ power to cover a given area\end{tabular}       & $\xmark$                                                                     & $\cmark$                                                                                   & $\xmark$                                                      & $\xmark$                                                            & $\cmark$                                                                    & $\xmark$                                                      \\ \hline
\cite{hayajneh2016drone}        & \begin{tabular}[c]{@{}l@{}}Quantify the impact of the number\\ and altitude of DBSs on the coverage\\ probability\end{tabular}                    & $\xmark$                                                                     & $\cmark$                                                                                   & $\xmark$                                                      & $\xmark$                                                            & $\xmark$                                                               & $\cmark$                                                     \\ \hline
\cite{3Dplacement}        & \begin{tabular}[c]{@{}l@{}}Optimal placement of a single DBS to\\ maximize the number of covered\\ ground users\end{tabular}                      & $\xmark$                                                                     & $\cmark$                                                                                   & $\xmark$                                                      & $\xmark$                                                            & $\cmark$                                                                    & $\cmark$                                                     \\ \hline
\cite{kalantari2016number}       & \begin{tabular}[c]{@{}l@{}}Optimal 3D placement and the number\\of required DBSs to cover a set of\\ ground users\end{tabular}                   & $\cmark$                                                                    & $\xmark$                                                                               & $\xmark$                                                          & $\xmark$                                                            & $\xmark$                                                               & $\cmark$                                                     \\ \hline
\cite{munaye2019uav}        & \begin{tabular}[c]{@{}l@{}}Optimal positioning of a single DBS\\ for users' throughput maximization\end{tabular}                                & $\cmark$                                                                    & $\xmark$                                                                               & $\xmark$                                                          & $\xmark$                                                            & $\xmark$                                                               & $\xmark$                                                     \\ \hline
\cite{liu2019reinforcement}        & \begin{tabular}[c]{@{}l@{}}Optima 3D placement and movement\\of DBSs for improving QoE\end{tabular}                                            & $\cmark$                                                                    & $\xmark$                                                                               & $\xmark$                                                          & $\cmark$                                                           & $\cmark$                                                                    & $\cmark$                                                     \\ \hline
\cite{mozaffariD2D}        & \begin{tabular}[c]{@{}l@{}}Optimal deployment of a single DBS\\ to maximize the sum-rate\end{tabular}                                           & $\cmark$                                                                    & $\cmark$                                                                                   & $\xmark$                                                      & $\xmark$                                                            & $\cmark$                                                                    & $\cmark$                                                     \\ \hline
\cite{kumbhar2019heuristic}        & \begin{tabular}[c]{@{}l@{}}Joint optimization of DBSs' location\\ and intercell interference coordination\\in LTE-advanced networks\end{tabular} & $\xmark$                                                                     & $\xmark$                                                                               & $\cmark$                                                          & $\cmark$                                                           & $\cmark$                                                                    & $\cmark$                                                    \\ \hline
\cite{namvar2019heterogeneous}        & \begin{tabular}[c]{@{}l@{}}Joint optimization of DBSs' location\\ and transmit power for providing\\ maximum coverage area\end{tabular}            & $\xmark$                                                                     & $\cmark$                                                                                   & $\xmark$                                                      & $\cmark$                                                           & $\xmark$                                                               & $\cmark$ \\ \hline
This work       & \begin{tabular}[c]{@{}l@{}}3D location optimization of DBSs for minimizing\\ the total transmit power while satisfying\\ the downlink rate requirement\end{tabular}           & $\cmark$                                                                     & $\cmark$                                                                              & $\cmark$                                                           & $\cmark$                                                            & $\xmark$                                                                   & $\cmark$                                                      \\ \hline
\end{tabular}
\end{table*}

The rest of this paper is organized as follows. Section~\ref{sec:literature} provides an overview of the recent state of the art for 3D placement of the DBSs. Section~\ref{sec:SySmodel} presents the system model and describes the air-to-ground channel model as well as the optimal flight altitude of each DBS as a function of their transmit power. The problem formulation is presented in Section~\ref{sec:ProblemFormulation}. The optimal selection and the deployment of the DBSs is investigated in Section~\ref{sec:subproblem1} while the interference management is addressed in Section~\ref{sec:subproblem2}. Numerical results are provided in Section~\ref{sec:simulations}. Finally, Section~\ref{sec:conclusion} concludes the paper and discusses the future path of this research.

\section{Literature Review}\label{sec:literature}
The envisioned opportunities for employing DBSs as a new tier for wireless networking has attracted remarkable recent research activities in the area. A substantial portion of the literature on DBSs is devoted to the AtG channel modeling. For instance, the authors in \cite{al2014modeling} provided a statistical generic AtG propagation model for Low Altitude Platform (LAP) systems in which the probability of LoS channel is derived as a function of the elevation angle. The work in \cite{feng2006path} studies the effects of shadowing and pathloss for UAV communications in dense urban environments. As discussed in \cite{holis2008elevation}, due to the pathloss and shadowing, the characteristics of the AtG channel depend on the height of the DBSs. A comprehensive survey on available AtG propagation models can be found in \cite{SurveyChannel}.

The 3D deployment of the DBSs is arguably the most influential design consideration in drone-based communications as it directly impacts the coverage, QoS, and life expectancy of the network \cite{mozaffari2019tutorial}. The optimal 3D placement of DBSs is a challenging task due to its dependency on the environmental factors (e.g., size and shape of the area), the AtG channel which itself is a function of DBS's altitude, and the location and/or distribution of the users on the ground. Consequently, the optimal deployment of DBSs has attracted considerable attention in the recent state of the art.

The optimal flight altitude of a single UAV-BS operating under the Rician fading channel is derived in \cite{Azari}. The authors in \cite{al2014optimal} developed an analytical framework to derive the optimal altitude of a single DBS, enabling it to achieve a maximum coverage radius on the ground. This result was extended to the case of two identical DBSs in \cite{mozaffariDesign}. The work in \cite{mozaffariLetter} investigated the problem of optimal 3D placement of a symmetric set of DBSs having the same transmit power and altitude. The authors in \cite{hayajneh2016drone} employed tools from stochastic geometry to analyze the impact of a DBS’s altitude on the sum-rate maximization. In \cite{3Dplacement}, a UAV-enabled small cell placement optimization problem is investigated in the presence of a terrestrial wireless network to maximize the number of users that can be covered. Furthermore, the authors in \cite{kalantari2016number} proposed a deployment plan for DBSs to minimize the number of drones required for serving the ground users within a given area. Similar works can be found in \cite{munaye2019uav,liu2019reinforcement,orfanus2016self, mozaffariD2D,kumbhar2019heuristic, merwaday2016improved, athukoralage2016regret}.

While these studies address important drone-based communication problems, they mainly limit their discussions to cases in which there exists only a single DBS or multiple identical DBSs with the same capabilities. Moreover, the number of DBSs to be deployed in a given area is assumed to be known in advance. In practice, however, one might have a repository of various types of drones with diverse capabilities in terms of flight altitude and transmit power. In this context, the exact number and the type of DBSs that need to be deployed depend on the target area and the number of ground users to be served. For instance, having a large set of DBSs, one may need to deploy only a few DBSs in order to cover a small area of interest. Otherwise, the efficiency of resource allocation may drop significantly due to over-allocation of resources. On the other hand, such over-allocation of resources may lead to excessive interference between the DBSs, which in turn, deteriorates the overall quality of service (QoS). In \cite{namvar2019heterogeneous}, the authors proposed a novel solution to handle the resource allocation and placement of the DBSs for a rectangular area of interest. However, the work in \cite{namvar2019heterogeneous} does not consider the location of the users and the placement is optimized to avoid interference between the DBSs. The assumptions of the related works are summarized in table \ref{tab:literature}.

\section{System Model}\label{sec:SySmodel}
Consider a heterogeneous repository of DBSs in which the DBSs are of various types depending on the range of their transmit power. For example, smaller drones may not afford high transmit powers, while larger drones, aerostats, and high altitude platforms, may support increasingly higher transmit powers.  Let $\mathcal{D} = \{D_i\}_{i = 1}^{N}$ denote the set of $N$ available drones in the repository while $P^t_i$ represents the transmit power of drone $D_i$. We further assume that $P_i^t\in[P_i^{\text{min}},P_i^{\text{max}}]$. There are $K$ mobile ground users distributed in a 2D geographical area with low to medium mobility. Let $\mathcal{T} = \{T_j\}_{j = 1}^{K}$ be the set of ground users. Our aim is to allocate available resources, i.e., the DBSs, to provide wireless coverage for the ground users with minimal aggregate transmit power. The type and the number of DBSs to be deployed depends on the number and the distribution of ground users.

Note that as we seek to optimize the allocation of the DBSs with minimum aggregate power, and we do not need to cover a region while there is no user there. As mentioned previously, one important feature of the DBSs is their ability to move and adapt their location to best serve the ground users. In this work, we find the placement of the DBSs for one snapshot of the users positions. The movement of the DBSs has been studied in \cite{mozaffari2016optimal,mozaffari2017wireless} by adopting the optimal transport theory framework. Moreover, the work in \cite{9037325} provides a good summary of the DBS trajectory optimization methods.

We further assume that the DBSs are connected via satellite links or long-range cellular backhaul links. Fig.~\ref{fig:sys} show the system model. In this paper, we consider drones as quasi-stationary low altitude platform (LAP). Note that although the LAP is quasi-stationary, the DBSs can hover at different altitudes to achieve their maximum possible coverage radius according to their transmit power. Moreover, the DBSs can reposition themselves to cope with the mobility of the ground users. In this regard, we seek to optimize the  location of the DBSs in order to provide wireless coverage with minimum energy consumption using the available drones in the repository.

\subsection{Air-to-Ground (AtG) Channel Model}
Selecting the proper AtG channel model is the crucial step in formulating the DBS placement problem. There are many empirical and analytical studies on AtG channel modeling in the literature. However, the majority of authors in this field have adopted the model presented in \cite{al2014modeling} as an accurate and convenient representation of AtG channel. A brief discussion of AtG is outlined in this section.

\begin{figure}
    \centering
    \includegraphics[width=9cm]{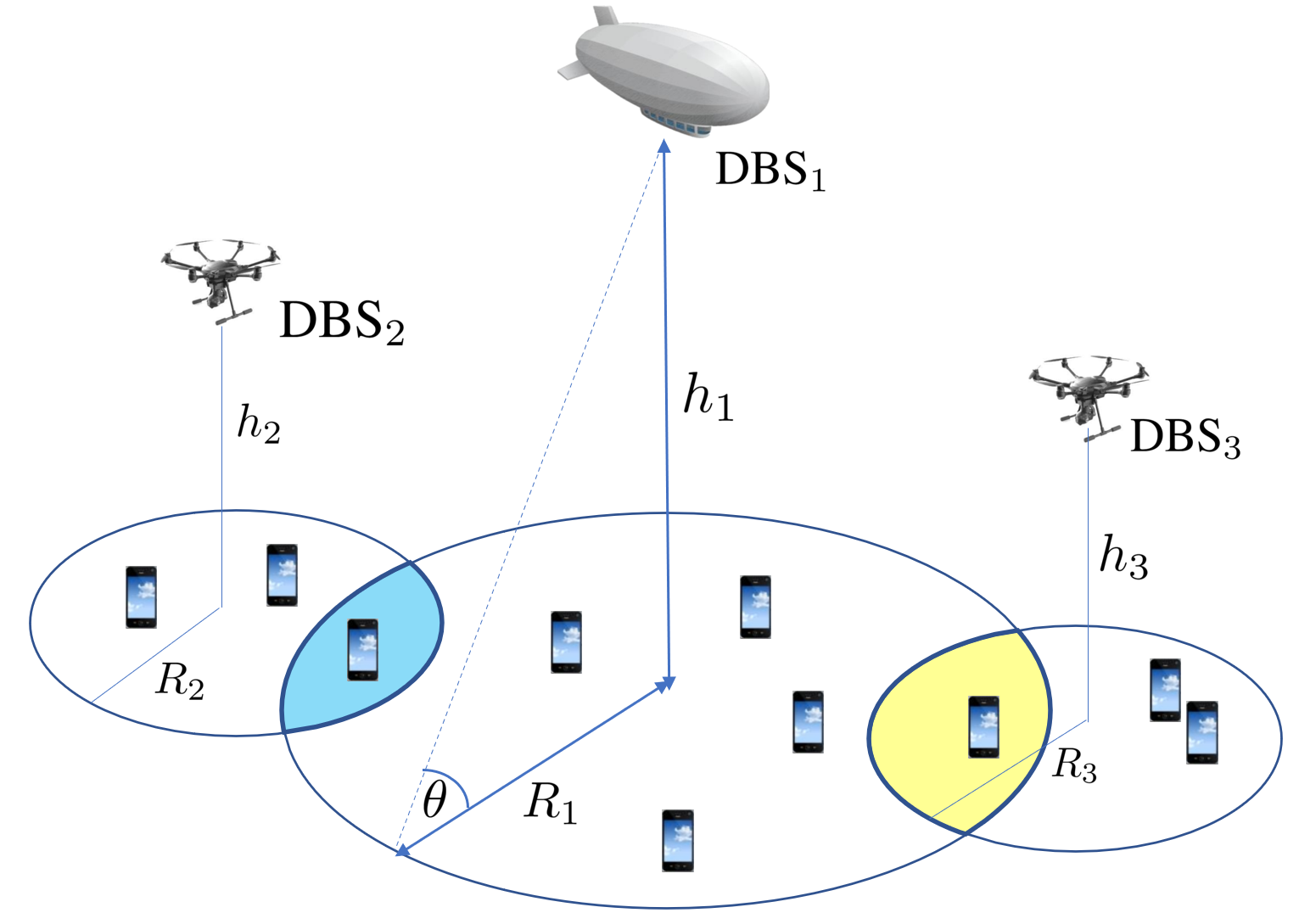}
    \caption{A heterogeneous set of DBSs with varying transmit power and altitude provide service for ground users. The overlapping areas undergo severe intercell interference on the downlink.} \vspace{-.3cm}
    \label{fig:sys}
\end{figure}

The radio signal from a LAP base station reaches its destination in accordance to two main propagation groups. The first group corresponds to receiving a LoS signal while the second group corresponds to receiving a strong non-LoS (NLoS) signal due to reflections and diffractions. These groups can be  considered separately with different probabilities of occurrence which depend on the environmental factors such as the density and height of buildings, and the elevation angle. In this work, we adopt the model presented in \cite{al2014modeling} for characterizing the AtG channels for LAP systems.

The AtG channel is modeled as a Bernoulli random variable for the LoS and NLoS paths. The corresponding probabilities of a LoS ($\Psi_\text{LoS}$) and NLoS ($\Psi_\text{NLoS}$ ) transmission between a transmitter and a receiver are given by:
\begin{align}
    \Psi _\text{LoS}&= \left[1+a\exp\left(-b\left(\frac{180 \theta}{\pi}-a\right)\right)\right]^{-1},\\
    \Psi _\text{NLoS}&= 1-\Psi _\text{LoS},
\end{align}
in which the constant parameters $\alpha$ and $\beta$ are determined by the environment, $\theta = \arctan(\frac{h}{r})$ is the elevation angle, $h$  is the altitude of DBS, and $r$ is the radial distance, respectively.

As it is difficult to determine whether a particular channel is LoS or NLoS, it is customary to consider the spatial expectation of the pathloss over LoS and NLoS links rather than the exact values of pathloss. The mean pathloss $\Gamma(\text{dB})$, is given by \cite{al2014modeling}:
\begin{equation}\label{expectedPL}
  \Gamma (\text{dB}) = \text{FSPL} + \eta_\text{LoS}\Psi _\text{LoS} + \eta_\text{NLoS}\Psi _\text{NLoS},
\end{equation}
where $\eta_\text{LoS}$ and $\eta_\text{NLoS}$ denote the excessive pathloss in LoS and NLoS links while $\text{FSPL} = 20\log\big(\frac{4\pi f_c d}{c}\big)$ is the free space pathloss, in which $f_c$ is the carrier frequency, $c$ is the speed of light, and $d = \sqrt{h^2+r^2}$ is the distance between the DBS and a ground point located at radial distance $r$.

By substituting $\Psi_\text{LoS}$ and $\Psi_\text{LoS}$ in (\ref{expectedPL}), we can see that $\Gamma$ is a function of $h$ and $r$, implying that the path loss is a function of the altitude and coverage of the DBS. Indeed, for a given $\Gamma$, the coverage of a DBS is a function of its altitude. The relationship between $\Gamma$, $h$, and $r$ is captured by the following:
\begin{equation}\label{expectedPL2}
   \Gamma= 20\log(d) + \frac{A}{1 + a \exp\big(-b(\theta-a)\big)} + B,
\end{equation}
in which $A = \eta_\text{LoS} - \eta_\text{NLoS}$ and $B =  \eta_\text{NLoS} + 20 \log\big(\frac{4\pi f_c}{c}\big)$.

\subsection{The Notion of Coverage and its Shape}
Having defined the expected pathloss in (\ref{expectedPL}), the received signal power at a ground receiver located in radial distance $r_i$ from the ground image of the DBS is given by
\begin{equation}\label{Preceived}
  P^{r} (\text{dB}) = P^t_i(\text{dB}) - \Gamma (\text{dB}).
\end{equation}

\begin{definition}
We define the service threshold in terms of the minimum allowable received signal power for a successful transmission. Any point in the area is covered if its received signal power is greater than a threshold $\epsilon$,
\begin{equation}\label{LossThreshold}
  \Gamma (\text{dB}) \leq P^t_i(\text{dB}) - \epsilon.
\end{equation}
\end{definition}

\begin{proposition}
For any given values of transmit power and flight altitude, the coverage area of a DBS is a circular disk.
\end{proposition}
\begin{proof}
{According to (\ref{LossThreshold}), for a given transmit power $P^t_i(\text{dB})$, the wireless coverage for a ground point only depends on the average pathloss $\Gamma (\text{dB})$ which is experienced in that point. However, the pathloss $\Gamma (\text{dB})$  in (\ref{expectedPL2}) is a function of a DBS's  distance to the ground station which is given by $d = \sqrt{h^2+r^2}$, in which $h$ is the altitude of DBS and $r$ is its the horizontal distance to the user. Hence, for a fixed hovering altitude $h$,  all the ground users at the radial distance $r$ experience the same pathloss. It is equivalent to saying that the locus of the points on the 2D area that experience the same pathloss is a circle centered at the ground image of the DBS. Thus, the coverage region of a DBS is a circular disk.}
\end{proof}

\begin{figure}
    \centering
    \includegraphics[width=9cm]{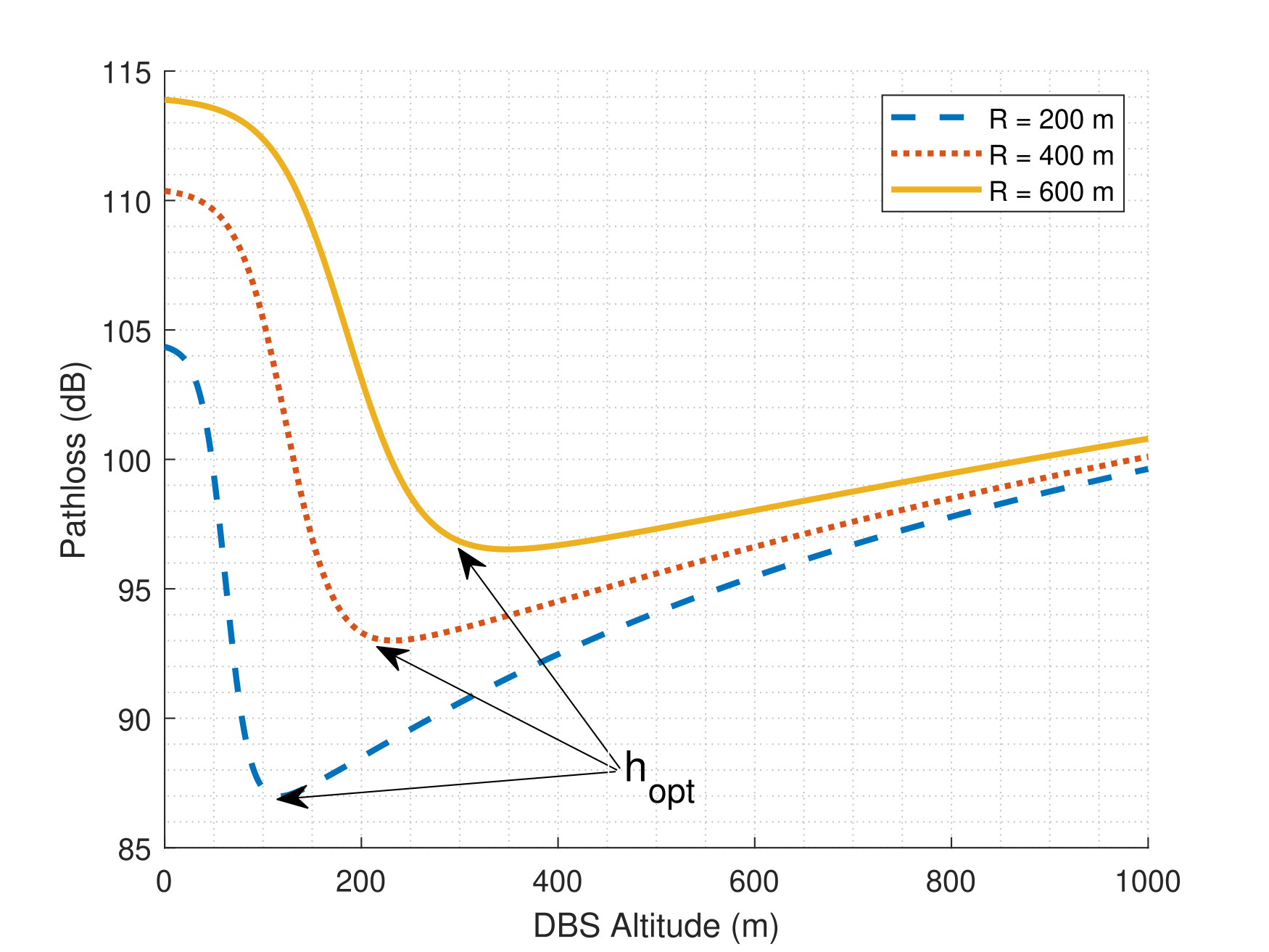}
    \caption{Pathloss as a function of DBS flight altitude for two fixed radial distances on the ground.} \vspace{-.3cm}
    \label{fig:pathloss}
\end{figure}

As shown in (\ref{expectedPL2}), the mean value of pathloss is an mplicit function of flight altitude. This function is shown in Fig.~\ref{fig:pathloss}. As it is seen in this figure, by increasing the altitude of a drone-BS, the pathloss first decreases and then increases. That is because in low altitudes the probability of NLoS is much higher than that of LoS, due to reflections by buildings and other objects, and the additional
loss of a NLoS connection is higher than a LoS connection, but when the altitude increases the LoS probability increases as well and in turn path loss decreases. On the other hand, the pathloss is also dependent on the distance between the transmitter and the receiver, so after a specific height, this factor dominates and as the altitude increases, the pathloss increases
as well. In a nutshell, the DBSs can be considered as a new tier of access nodes in cellular communication systems where the desired coverage area can be attained by changing the transmit power and/or the flight altitude. Moreover, as shown in Fig.~\ref{fig:pathloss}, for a fixed value of radial distance, the pathloss is a unimodal function of altitude. This observation leads us to find the optimal flight altitude of a DBS in the next section.

\subsection{Optimal Flight Altitude for a Single DBS}
The coverage radius of a DSB with transmit power $P^t$ is defined as the radial distance in which the received signal power on a ground receiver reaches the threshold $\epsilon$, i.e.,
\begin{equation}\label{Rdefinition}
    R = r\vert_{\Gamma = P^t - \epsilon},
\end{equation}
in which $R$ is the coverage radius of the DSB. By substituting (\ref{expectedPL2}) into the definition above, we have,
\begin{equation}\label{Requation}
   20\log(d) + \frac{A}{1 + a \exp\Big(-b \big[\arctan(\frac{h}{R})-a \big]\Big)} + B +\epsilon = P^t,
\end{equation}
in which $A = \eta_\text{LoS}(\text{dB}) - \eta_\text{NLoS}(\text{dB})$ and $B =  \eta_\text{NLoS}(\text{dB}) + 20 \log(\frac{4\pi f_c}{c})$. The equation (\ref{Requation})
shows that for any given value of $P^t$, the radius $R$ is an implicit function of $h$. As shown in Fig.~\ref{fig:pathloss}, this function (i.e., $f(h,r)\vert_{P^t=\text{cte}}=0$) is unimodal. As a unimodal function, $f(h,r)\vert_{P^t=\text{cte}}=0$ has only one stationary point which corresponds to the maximum coverage radius. In order to find this stationary point, we take the partial derivative $\frac{\partial r}{\partial h} = 0$, which can be expanded to:
\begin{equation}\label{thetaoptimal}
    \frac{h}{R}+ \frac{9\ln(10)ab A\exp\left(-b[\arctan\big(\frac{h}{R}\big)-a]\right)}{\pi\left[a\exp\left(-b[\arctan\big(\frac{h}{R}\big)-a]\right)+1\right]^2}=0.
\end{equation}

We can find the optimal flight altitude with the corresponding coverage radius by solving the simultaneous equations (\ref{Requation}) and (\ref{thetaoptimal}). There is no closed-form solution to these simultaneous equations and we need to resort to numerical methods to find the optimal values of $h$ and $R$.

Note that due to the practical limitations on DBS altitude, we have $h\leq h_{\text{max}}$, where $h_{\text{max}}$ is the maximum allowable flight altitude in the given environment. As it is shown in Fig.~\ref{fig:pathloss}, for any given coverage radius, the DBS transmit power $\Gamma$ decreases with increasing the altitude of DBS up to some point, i.e. $h_{\text{opt}}$, and then increases. Thus, considering the imposed limitation on the DBS flight altitude, the feasible optimal flight altitude which leads to minimal power consumption for a given coverage radius is equal to $\hat{h}_{\text{opt}} = \min \{h_{\text{opt}}, h_{\text{max}} \}$.

\section{Problem Formulation} \label{sec:ProblemFormulation}
Considering the system model in Fig.~\ref{fig:sys}, we investigate the joint problem of selection and the 3D placement for a heterogeneous set of of DBSs to provide wireless coverage for the ground users. After the locations of all DBSs are determined, each ground user is associated with the DBS that has the highest SINR. We consider the transmission between DBS $i$ and a ground user located at $(x,y)$ coordinates. The achievable rate for the user is given by:
\begin{equation} \label{rategamma}
\gamma_{i}(x, y)=W_{i}\log_{2}\left(1+\frac{P_{i}(x, y)/\Gamma_{i}(x, y)}{N_{0} + \sum_{ j\neq i}^{N} P_{j}(x, y)/\Gamma_{j}(x, y) }\right),
\end{equation}
where $W_{i}$ is the transmission bandwidth of DBS $i$, $P_i(x, y)$ is the DBS transmit power to the user, $\Gamma_{i}(x, y)$ is the average path loss between DBS $i$ and the user, and $N_{0}$ is the noise power. Clearly, the number of users covered by the DBS depends on the distribution of users and the location of the DBS.

The minimum transmit power required to satisfy the rate requirement $\beta$ of ground users is given by:
\begin{equation}\label{minPower}
P_{i,\min}(x,y)=
\left(2^{\beta/W_{i}}-1\right)\Gamma_{i}(x, y)\big(N_{0} + \sum_{ j\neq i}^{N} \frac{P_{j}(x, y)}{\Gamma_{j}(x, y)}\big),
\end{equation}
which is derived using (\ref{rategamma}) and $\gamma(x,y)>\beta$.


To provide the maximum coverage in a geographical region which includes a number of ground users with known locations with the minimum total transmit power, one needs to answer the following questions:
\begin{itemize}
\item How many and which types of the DBSs should be selected?
\item For any subset of the DBSs, what is the optimal placement to achieve minimum aggregate transmit power?
\item In case of overlapping coverage areas, how to address the inter-cell interference?
\end{itemize}
These questions can be formulated as the following optimization problem:
\begin{align}
&\underset{I_i, (x_i, y_i, h_i)}{\text{minimize}}  \quad \sum_{i = 1}^{N}I_i  P^t_i(x_i, y_i), \label{Optimization} \\
&\text{s.t.} \notag\\
& I_i \in \{0,1\},   \label{constraint1} \\
& \gamma_k(x,y)\geq \gamma_0, \label{constraint2} \\
& P_i^t\geq P_{i,\min}, \label{constraint3}\\
& h_i\leq h_i^{\text{opt}}, \label{constraint4}
\end{align}
where $N$ is the total number of available UAVs in the repository. In addition, $I_i$ is an indicator function which equals to $1$ if DBS $D_i \in \mathcal{D}$ is selected for covering the region and equals to $0$ otherwise. It governs the resource allocation strategy for a given area of interest. Moreover, $\gamma_k (x,y)$ is the downlink transmission rate of user $k$ at location $(x,y)$ on the ground which is a function of SINR. Note that a user is assigned to the nearest UAV as is the case for the terrestrial networks.

The first constraint in (\ref{constraint1}) governs the DBS selection scheme. The second constraint in (\ref{constraint2}) ensures the quality of service for the ground users. The third and fourth constraints in (\ref{constraint3}) and (\ref{constraint4}) control the transmit power and flight altitude of the DBSs. In particular, constraint (\ref{constraint4}) highlights the fact that for the flight altitude range $h_i\leq h_i^{\text{opt}}$, the transmit power is an increasing function of the coverage radius, as discussed earlier. This constraint restricts the transmit power to be a monotonic function of coverage radius and allows a tractable solution for the optimization problem in (\ref{Optimization}).

\subsection{Methodology}
Due to its non-convexity, non-linear constraints, and the large number of unknowns, the optimization problem stated in (\ref{Optimization}) is very challenging to solve. We divide the optimization problem in (\ref{Optimization}) into two sub-problems and solve them sequentially to find the best subset of available DBSs along with their corresponding 3D position.

In the first sub-problem, we neglect the constraints (\ref{constraint2}) and (\ref{constraint3}), i.e., the achievable rate performance, and find a subset of DBSs to cover the ground users with minimal transmit power. This problem is similar to the so-called disk covering problem \cite{suzuki1996p}, but with some substantial differences as described in Section \ref{sec:subproblem1}. Nonetheless, we can employ some ideas from the disk covering problem to devise a proprietary solution for the problem in hand.

In the second sub-problem, given the selected DBSs and their corresponding 3D locations, we introduce a low-complexity beamforming method to alleviate the co-channel interference in the overlapping areas between two DBSs where the impairment caused by the interference is most severe. In this problem, our goal is to achieve the required transmission rate at the ground users within the resulting topology from the first sub-problem.

There exists an interplay between these two problems as the solution of the first problem impacts the solution of the second problem and vice-versa. Thus, we propose a recursive algorithm to solve these two interdependent problems. We also provide a time complexity analysis of the proposed algorithm as a measure of its efficiency and scalability.

\section{Selection and 3D placement of the DBSs}\label{sec:subproblem1}
In this section, we investigate the joint problem of resource allocation and optimal placement of a heterogeneous set of DBSs. Consider $K$ ground terminals $\mathcal{T} = \{T_j\}_{j = 1}^{K}$ that are distributed in a two-dimensional area and let $(\tilde{x}_j, \tilde{y}_j)$ be the coordinates of ground terminal $T_j$. Moreover, let $(x_{i}, y_{i}, h_{i})$ denote the three-dimensional location of DBS $D_i$  with transmit power $P_i^t$. We need to solve the following optimization problem,
\begin{align}
&\underset{I_i, (x_{i}, y_{i}, h_{i})}{\min}  \quad \sum_{i = 1}^{N}I_i  P^t_i(x_{i}, y_{i}, h_{i}), \label{Optimization2} \\
&\text{s.t.} \notag\\
& I_i \in \{0,1\},   \label{constraint2-1} \\
& P^r(\tilde{x}_j, \tilde{y}_j)\geq \epsilon, \qquad \forall T_j\in \mathcal{T} \label{constraint2-2} \\
& h_{i} = \min \{h_{i_{\rm opt}}, h_{\rm{max}} \},  \label{constraint2-3}\\
& P_i^t\in[P_{i_{\rm min}},P_{i_{\rm max}}], \label{constraint2-4}
\end{align}
in which the constraint (\ref{constraint2-1}) controls  the  DBS  selection strategy while the constraint (\ref{constraint2-2})  guarantees that all the ground users are covered with \textit{at least} one DBSs. Recall that the received power at ground terminal $T_j$ from DBS $D_i$ is given by $P^r(\tilde{x}_j, \tilde{y}_j) = P^t_i(x_{i}, y_{i}, h_{i}) - \Gamma(d_{ij})$ in which the mean pathloss $\Gamma(d_{ij})$ is a function of the distance between $T_j$ and $D_i$, which is  $d_{ij} = \sqrt{(x_{i}-\tilde{x}_j)^2+(y_{i}-\tilde{x}_j)^2+h_{i}^2}$. Finally, the constraint in (\ref{constraint2-3}) ensures that the DBSs hover at their optimal altitude and the constraint in (\ref{constraint2-4}) limits the transmit power of the DBSs.

It is worth noting that restricting the DBS flight altitude to $h\in (0, h_{\text{opt}}]$, causes the coverage radius $R$ to be an increasing function of transmit power $P^t$. Therefore, given the fact that the flight altitude of DBSs is restricted to the range $(0, h_{\text{opt}}]$ by the constraint (\ref{constraint2-3}), we can minimize the DBSs' coverage radii instead of their transmit power. We will do so by minimizing the coverage radii of DBSs one at a time, starting from the largest coverage radius. Consider the following optimization problem which aims at minimizing the largest coverage radius of the deployed DBSs:
\begin{align} \label{K-center}
    & \underset{(x_{i}, y_{i})\vert_{i=1}^{N}}{\min}\quad
    \underset{(\tilde{x}_j, \tilde{y}_j)\vert_{j=1}^{K}}{\max}\quad
    \underset{(x_{i}, y_{i})\vert_{i=1}^{N}}{\min}I_iD_{ij}
    \\
    &\text{s.t.} \notag\\
    & I_i \in \{0,1\}   \label{K-center-constraint}
\end{align}
in which $D_{ij} = \sqrt{(x_{i}-\tilde{x}_j)^2+(y_{i}-\tilde{x}_j)^2}$ is the radial distance between ground terminal $T_j$ and the image of DBS $D_i$ on the two-dimensional Cartesian plane and the constraint (\ref{K-center-constraint}) controls the subset selection of the available DBSs.

The optimization problem in (\ref{K-center}) aims at arranging the $M\le N$ coverage disks in the two-dimensional plane such that: 1) all the ground terminals are covered by at least one coverage disk; and, 2) the radius of the largest coverage circle is minimized. Let us assume that $K_1<K$ ground users are covered by the largest disk. Once the solution to (\ref{K-center}) is found, we can remove the largest coverage disk and all its encircled ground terminals and solve the smaller version of the same problem with $M-1$ coverage disks and $K-K_1$ ground terminals. This recursive algorithm continues until there remains only one disk whose optimal placement is determined for covering the remaining ground terminals with minimum radius. Next, we propose an efficient algorithm to solve the optimization problem in (\ref{K-center}).

\subsection{Proposed Algorithm}
The problem in (\ref{K-center}) has an analogy with the so called \textit{planar K-center problem} \cite{suzuki1996p}. In the \textit{K}-center problem, for a given set of points on a two-dimensional surface, the task is to arrange a given number of congruent disks, say $M$ disks, centered at $M$ points from the set, such that all the points on the surface are covered. The goal is to minimize the radius of these congruent disks.  The problem is known to be NP-hard \cite{plastria2002continuous}, and hence, there does not exist a polynomial time algorithm to solve it optimally. There are many studies in the literature to tackle the different variants of the \textit{K}-center problem, most of which are heavily influenced by the constraint that disk centers have to be selected from the set of points (e.g., the users' location in our problem) \cite{banhelyi2015optimal}. Nevertheless, the optimization problem in (\ref{K-center}) has some distinctive differences with the planar \textit{K}-center problem which mandates devising a solution tailored to the specific properties  of the problem in hand. First, unlike the original \textit{K}-center problem, the center of the coverage disks can be chosen anywhere in the surface, and not restricted to some predetermined points. Second, different from the \textit{K}-center problem, the number of coverage circles are not known a priori. Finally, as the coverage radii of the coverage disks are mere translation of their corresponding DBS transmit power, there exists a limitation on the maximum and minimum coverage radii of the coverage disks. In other words, not all the solution to (\ref{K-center}) are feasible and we need to find a feasible optimal solution.

First let us define some concepts that we will need later. Let $\Tilde{\mathcal{T}}$ be a subset of $\mathcal{T}$ and let $F(\Tilde{\mathcal{T}})$ be the optimal value of the objective function (\ref{K-center}) with a single coverage disk that covers all the ground terminals in $\Tilde{\mathcal{T}}$. This problem is known as the planar 1-center problem which is modeled as a linear programming problem and can be solved in $\mathcal{O}(n)$ time \cite{megiddo1983linear}. $F(\Tilde{\mathcal{T}})$ is given by:
\begin{equation}\label{1-center}
    F(\Tilde{\mathcal{T}}) = \underset{(x, y)}{\min}\quad
    \underset{(\tilde{x}_j, \tilde{y}_j)\in \Tilde{\mathcal{T}}}{\max} \sqrt{(x-\tilde{x}_j)^2+(y-\tilde{y}_j)^2}
\end{equation}
In fact, $F(\Tilde{\mathcal{T}})$ is the radius of the smallest disk that covers all the points in $\Tilde{\mathcal{T}}$. Also, let $X^*(\Tilde{\mathcal{T}}) = (x^*,y^*)$ be the optimal point of optimization problem (\ref{1-center}). It is shown in \cite{farahani2009facility} that $X^*(\Tilde{\mathcal{T}})$ is unique for any set of points $\Tilde{\mathcal{T}}$.

Next, let $\varphi = \{ \Tilde{\mathcal{T}}_1, \Tilde{\mathcal{T}}_2,\dots, \Tilde{\mathcal{T}}_M\}$ be a partition of $\mathcal{T}$ such that $\Tilde{\mathcal{T}}_i \bigcap \Tilde{\mathcal{T}}_j = \emptyset$ and $\bigcup_{i=1}^{M}\Tilde{\mathcal{T}}_i = \mathcal{T}$. Also, let $F_\varphi$ be the optimal value of the objective function for partition $\varphi$ which is given by,
\begin{equation}
    F_\varphi =  \left\{ F(\Tilde{\mathcal{T}}_1), F(\Tilde{\mathcal{T}}_2),\dots, F(\Tilde{\mathcal{T}}_M)\right\},
\end{equation}
which is the set of radii of the smallest disks to cover all the points of $\mathcal{T}$ according to partition $\varphi$.


We select $M$ starting points $\{X^{[0]}_{1}, X^{[0]}_{2}, \dots, X^{[0]}_{M}\}$ as the initial centers of the coverage disks where $X^{[0]}_{i} = (x^{[0]}_{i}, y^{[0]}_{i})$. Let $k$ denote the iteration number and suppose that $\{X^{[k]}_{1}, X^{[k]}_{2}, \dots, X^{[k]}_{M}\}$ are found by solving the 1-center problem \cite{megiddo1983linear} for the corresponding subsets in partition $\varphi^{[k]}$. A set of ground terminals is assigned to each center such that each ground terminal is assigned to the closest center. Consequently, the $M$ centers in iteration $k$ define the following partitioning:
\begin{multline} \label{NewPartitions}
 \Tilde{\mathcal{T}}_i^{[k]} = \Big\{ T_p\big\vert \quad \sqrt{(x^{[k]}_{i}-x_{T_p})^2+(y^{[k]}_{i}-y_{T_p})^2} < \\
 \sqrt{(x^{[k]}_{j}-x_{T_p})^2+(y^{[k]}_{j}-y_{T_p})^2}, \quad j = 1,2,\dots, M \Big\}.
\end{multline}

Next, we update the centers of the disks according to the new partitions in (\ref{NewPartitions}). The new center for each set is the solution to the 1-center problem defined by the following set,
\begin{equation} \label{NewCenters}
  X^{[k+1]}_{i} =
    \begin{cases}
      X^*(\Tilde{\mathcal{T}}^{[k]}_{i}) & \text{if $\Tilde{\mathcal{T}}^{[k]}_{i}\neq\emptyset$}\\
      X^{[k]}_{i} & \text{if $\Tilde{\mathcal{T}}^{[k]}_{i}=\emptyset$}\\
    \end{cases}.
\end{equation}
The algorithm runs until the partitioning does not change, i.e., $\varphi^{[k]} = \varphi^{[k+1]}$, for some $k$. Upon convergence, the proposed algorithm partitions the ground terminal into $M$ subsets and for each subset it yields the optimal location of the coverage disk with minimum radius that covers all the ground terminal in that subset.

Given the coverage radii of the DBSs, their corresponding transmit power and flight altitude are found by solving (\ref{Requation}) and (\ref{thetaoptimal}). However, a solution to (\ref{Optimization2}) is feasible only if the calculated transmit power of all selected DBSs fall within the allowed range according to the constraint~(\ref{constraint2-4}). Let $\mathcal{S}$ be the set of all feasible solution defined. Upon finding all the feasible solutions, we select the one with the minimum aggregate transmit power as the optimal solution. Algorithm~\ref{Alg:3Dplacement} shows the pseudocode of the proposed algorithm.

It is worth noting that some ground terminals may be covered by more than a single coverage disk, which results in experiencing strong co-channel interference. Next, we present a bargaining game formulation to model and alleviate the co-channel interference for overlapping DBSs.

\begin{algorithm}[t]
\SetAlgoLined
\SetKwInOut{initialization}{Initialization}
\KwData{$\mathcal{D} = \{D_i\}_{i = 1}^{N}$, $\mathcal{T} = \{t_j\}_{j = 1}^{K}$, $h^{\text{max}}$}
\KwResult{\{$I_i, (x_{D_i}, y_{D_i}, h_{D_i})\}_{i = 1}^{N} $ }
\initialization{$I_i \leftarrow 0$ for $i = 1,2,\dots,N$ \newline $\mathcal{S}$ $\leftarrow \emptyset$}

\For{$M \leftarrow 1$ \KwTo $N$}{
crt $\leftarrow 0$\\ \tcc{Loop counter}
Initialize $M$ centers $\{X^{[0]}_{1}, X^{[0]}_{2}, \dots, X^{[0]}_{M}\}$\\
\Repeat{$\varphi[crt] = \varphi[crt+1]$}{
calculate the partition $\varphi[crt]$ using (\ref{NewPartitions})\\
crt $\leftarrow$ crt + 1\\ 
update the centers using (\ref{NewCenters})\\
calculate the partition $\varphi[crt+1]$ using (\ref{NewPartitions})\\
}

\For{$j \leftarrow 1$ \KwTo $M$}{
$R_j \leftarrow F(\Tilde{\mathcal{T}}_j)$\\
calculate the corresponding $P^t_j$ using (\ref{Requation}),  (\ref{thetaoptimal})\\
calculate the corresponding $h_{\text{opt},j}$ using (\ref{Requation}),  (\ref{thetaoptimal})\\
\For{$i \leftarrow 1$ \KwTo $N$}{
\If{$P^t_j \in [P_i^{\text{min}}, P_i^{\text{max}}]$}{
$I_i \leftarrow 1$\\
$R_i \leftarrow R_j$\\
$h_i \leftarrow \min \{h_{\text{opt},j}, h^{\text{max}} \}$\\
$\mathcal{D} \leftarrow \mathcal{D}\setminus D_i$

}
}
}
\If{$\sum_{i=1}^{N}I_i = M$}{
$\mathcal{S}\leftarrow \mathcal{S} \bigcup \big\{(x_i,y_i,h_i,R_i,P^t_i)\big\vert I_i = 1\big\}$
}
}
\KwRet{ $\mathcal{S}$}
\caption{3D Placement of DBSs}
\label{Alg:3Dplacement}
\end{algorithm}

\section{Downlink Beamforming for Interference Management}\label{sec:subproblem2}
Enabling airborne adhoc systems to efficiently operate in the same spectral band is a key challenge for the drone small cells. Similar to the terrestrial wireless cells, the inter-cell interference caused by communication in an interference channel degrades the quality of received signals on the ground stations. There are many algorithms and solutions to alleviate the impact of destructive co-channel interference for the conventional terrestrial networks \cite{mhiri2013survey}. However, the difficulty in the airborne small cell networks stems from the fact that the DBSs are very battery-limited and thus, implementing the conventional interference management methods for DBSs is not a viable option due to overwhelming computational complexity. In this section, we adopt the framework of \textit{bargaining game theory} \cite{peters2013axiomatic} to introduce a simple and low-complexity beamforming method to address the inter-cell interference.

We consider the scenario whereby $M$ interfering DBSs are trying to transmit their information in the downlink to $M$ ground users located in the overlapping region of their corresponding coverage disks. Assuming that each DBS performs single-stream transmission, and given that all channels are frequency flat, we have the following complex baseband symbols $y_m$ received by the ground users $T_m$:
\begin{equation}\label{Y1}
    y_{m}=  {h}_{mm}^{T}  {w}_{m} s_{m} + \sum _{l=1,l\ne m}^{M}  {h}_{lm}^{T} {w}_{l} s_{l} + e_{m},
\end{equation}
where $s_m, 1<m<M$ is the transmitted symbols from DBS $D_m$, $h_{lm}$ represents the $K \times 1$ channel vector between DBS $D_l$ and user $T_m$, $w_m$ is the $k \times 1$ beamforming vector used by DBS $D_m$, and  $n_m$ is the zero mean additive Gaussian noise with variance $\sigma^2$. The maximum transmit power per DBS is normalized to $1$ which yields the following power constraint on each DBS $D_m$: $\parallel w_m \parallel^2 \leq 1, \forall m \in\{1, 2, \dots, M\}$.

Each DBS $D_m$ seeks to optimize its weight vector $w_m$ in order to maximize the quality of service received by its corresponding ground user. However, there exists an interplay between the strategies (i.e., optimizing the weight vectors) of the DBSs as any selected value of $w_m$ impacts the choice of $w_l, l\ne m$, and vice versa. Thus, the key question is whether we can implement some sort of cooperation between the interfering DBSs to improve their performance? To answer this question, we first assess the non-cooperative scenario in which the DBSs act selfishly with no exchange of information. We formulate a non-cooperative \textit{zero-sum game} between the interfering DBSs for which the \textit{Nash equilibrium} is the accepted outcome \cite{han2012game}.
\subsection{Nash equilibrium}
Consider the non-cooperative downlonk beamforming game $G$ as the triplet $G = \big\{\mathcal{M}, ({S}_m)\vert_{m \in \mathcal{M}}, (u_m)\vert_{m \in \mathcal{M}} \big\}$ where:
\begin{itemize}
\item $\mathcal{M}$ is the set of players, i.e., the interfering DBSs;
\item ${S}_m$ is the strategy of DBS $D_m$ which is its choice of weight vector $w_m$ such that $\parallel w_m \parallel^2 \leq 1$;
\item $S_{-m}$ is the vector of strategies of all DBSs except $D_m$;
\begin{equation*}
    S_{-m} = [S_1, \dots, S_{m-1}, S_{m+1}, \dots, S_M];
\end{equation*}
\item $u_m: [S_m, S_{-m}] \rightarrow \mathbb{R}$ is the utility of each DBS $D_m$ which is the rate it achieves at its correponding ground user.
\end{itemize}

For a given tuple of beamforming vectors $(w_1,w_2, \dots, w_M)$, the received rate at the ground users is given by:
\begin{equation}\label{R1}
    R_m = \log_2 \left( 1 + \frac{\vert w_m^Th_{mm} \vert^2}{\sigma^2 + \sum _{l=1,l\ne m}^{M}
    \vert w_l^Th_{lm} \vert^2} \right).
\end{equation}
We define the utilities of the DBSs as
\begin{equation}\label{utility}
    u_m (S_m, S_{-m}) = R_m (w_1,w_2, \dots, w_M).
\end{equation}

As the utilities depend on the strategies of the competing players, we have a noncooperative game among the DBSs. In the absence of coordination among the DBSs, the outcome of the game will generally be the Nash equilibrium. A vector of strategies $(S_1^{\rm NE}, S_2^{\rm NE}, \dots, S_M^{\rm NE})$ is the Nash equilibrium if it satisfies the following condition:
\begin{equation}\label{Nashcondition1}
    u_m(S_m^{\rm NE}, S_{-m}^{\rm NE})\geq u_m(S_m, S_{-m}^{\rm NE}), \quad 1\leq m \leq M,
\end{equation}
which means that no DBS can unilaterally deviate from its optimal Nash equilibrium strategy without decreasing its own utility. By substituting (\ref{R1}) and (\ref{utility}) in (\ref{Nashcondition1}) and by performing some algebraic manipulations, we can find the unique equilibrium strategies as,
\begin{equation}\label{NashW1}
    w_m^{\rm NE} = \frac{h_{mm}^\ast}{\vert \vert h_{mm} \vert \vert}, \quad 1\leq m \leq M,
\end{equation}
where $h_{ij}^\ast$ is the complex conjugate of $h_{ij}$. The equilibrium strategies in (\ref{NashW1}) correspond to the maximum-ratio transmission beamforming. This conclusion is resulted from the fact that when DBS $D_m$ uses the beamforming vector $w_m^{\rm NE}$
at the Nash equilibrium, there exists no other vector that can yield a larger rate while satisfying the power constraint $\parallel w_m \parallel^2 \leq 1$.

\subsection{Bargaining Solution}
The Nash strategies in (\ref{NashW1}) represent the natural outcome of the considered scenario and it does not necessarily amount to the optimal beamforming strategies for the DBSs. In fact, the inefficiency of the Nash equilibrium solution in (\ref{NashW1}) is due to the uncoordinated actions of the DBSs. Our goal is to improve on the Nash strategies by allowing some level of cooperation among the DBSs. Note that the DBSs are selfish and try to maximize their own individual rate. Therefore, any sort of cooperation is feasible only if leads to better utilities, i.e., downlink transmission rate, for all the involved parties. We show that by small exchange of information between the interfering DBSs and without any need for a centralized controller, they can coordinate their strategies such at all the involved DBSs benefit from the cooperation.

We define a bargaining game between the interfering DBSs in which the DBSs need to find a point in the achievable rate region which yields better individual transmission rate for all the interfering DBSs compared to the non-cooperative scenario, i.e., the Nash equilibrium rates. Once they agree on a point, the beamforming are optimized accordingly. Let's define the achievable rate region $\overline{\mathcal{R}}$ as
\begin{equation}\label{achievablerate}
   \overline{\mathcal{R}} =  \underset{w_i,1\leq i \leq M, \parallel w_i \parallel \leq 1}{\mbox{\LARGE$\cup$}}(R_1, R_2,\dots, R_M),
\end{equation}
which is a compact set since the set $\{w_i\}_{i=1}^M$ subject to power constraint $\parallel w_i \parallel \leq 1$ is compact and the mapping from $\{w_i\}_{i=1}^M$ to $\{R_i\}_{i=1}^M$ is continuous.

The bargaining game between $M$ DBSs with interfering regions is composed of the following elements:
\begin{itemize}
    \item The compact and convex set $\mathcal{R}$ of all possible utilities (i.e., transmission rates) of the players (i.e., the DBSs). Note that the achievable rate region $\overline{\mathcal{R}}$ defined in (\ref{achievablerate}) is not necessarily a convex set, thus, we consider the convex hull $\mathcal{R}$ of the points in $\overline{\mathcal{R}}$ as the  bargaining utility region $\mathcal{R}$. Formally, $\mathcal{R} = \textbf{Conv}(\overline{\mathcal{R}}) $.
    \item The \textit{disagreement point} $d$ which is the outcome of the bargaining game if the players fail to reach an agreement. Let the disagreement point be the Nash equilibrium rates that the players achieve in absence of cooperation:
\begin{equation} \label{diagreementPoint}
\begin{split}
d & = \big( R_1, R_2,\dots, R_M \big)\vert_{w_i = w_i^{\rm NE}, 1\leq i \leq M} \\
 & = \big( R_1^{\text{NE}}, R_2^{\text{NE}},\dots, R_M^{\rm NE} \big).
\end{split}
\end{equation}
\end{itemize}

\begin{definition}
The bargaining solution is a function $f(.)$ that specifies
a unique outcome $f(\mathcal{R}, d) \in \mathcal{R}$ for every bargaining problem $(\mathcal{R}, d)$. Let $f_i(\mathcal{S}, d)$ represent the component of player $i$ in the bargaining outcome.
\end{definition}

\begin{definition} A bargaining solution $f(\mathcal{R}, d)$ is Pareto efficient if there does not exist a point $(R_1,R_2, \dots, R_M)\in \mathcal{M}$ such that $R>f(\mathcal{R},d)$ and $R_i>f_i(\mathcal{R},d)$ for some $i$. Any reasonable bargaining scheme must choose a Pareto efficient outcome since, otherwise, there would exist another outcome which is better off for all the players.
\end{definition}

There exist several axiomatic definitions for the solution of bargaining game, such as the Nash bargaining solution (NBS) \cite{han2012game} and the Kalai–Smorodinsky bargaining solution (KSBS) \cite{kalai1975other}. In this work, we use the KSBS solution concept since it results in individual fairness which will be explained in the remainder of this section. In the KSBS formulation, given the beamforming bargaining problem $(\mathcal{R}, d)$, the rate for the DBSs satisfies the following equation,
\begin{equation}\label{KSformula}
   \frac{R_{1}-R_{1}^{\text{NE}}}{R_{1}^{\text{max}}-R_{1}^{\text{NE}}} = \cdots = \frac{R_{M}-R_{M}^{\text{NE}}}{R_{M}^{\text{max}}-R_{M}^{\text{NE}}},
\end{equation}
where $R_{m}^{\text{NE}}$ is the rate of player $m$ at the disagreement point (i.e, Nash equilibrium), and $R_{m}^{\text{max}}$ denotes the maximum possible rate for player $m$. For our problem, achieving $R_{m}^{\text{max}}$ corresponds to allowing user $m$ to occupy all available resources, i.e., all the bandwidth in interference channel, and thus it is easy to determine. Thus, in the Kalai–Smorodinsky bargaining solution (\ref{KSformula}), every player gets the same fraction of its maximum possible rate which makes it an attractive approach in situations where one wishes to balance individual fairness with overall system performance. The optimal values of transmission rate according to KSBS can be found by solving the following optimization problem
\begin{align}\label{KS2}
&\underset{(w_1, w_2, \dots, w_M)}{\text{maximize}} \qquad r \\
&\text{s.t.} \quad r = \frac{R_{m}}{R_{m}^{\text{max}}}, \quad i = 1,2,\dots, M.
\end{align}


The optimization problem in (\ref{KS2}) cannot be solved by the convex optimization techniques due to the fact that the additional equality constraints in (36) are not affine with respect to $r$ and $(w_1, w_2, \dots, w_M)$. However, inspecting (\ref{KSformula}) and (\ref{KS2}) reveals that the KSBS corresponds to the intersection of the rate region boundary and the line segment from the origin to the point $(R_{1}^{\text{max}},\dots,R_{M}^{\text{max}})$. Thus, the optimization problem (\ref{KS2}) can effectively be solved by employing the \textit{bisection} method \cite{boyd2004convex}. The goal of the bisection method is to efficiently search along this line segment until the point of intersection is found.

Since, $0\leq r \leq 1$, we set $B_l = 0$ and $B_u = 1$ as the lower bound and upper bound of $r$. At each iteration, we bisect the the interval between $B_l$ and $B_u$ by the midpoint $\frac{B_l+B_u}{2}$. Then, we check the feasibility test
\begin{equation}\label{feasibility}
    \frac{R_{m}}{R_{m}^{\text{max}}}>r', \quad m = 1,\cdots, M,
\end{equation}
to determine if the current bisection point $r' = \frac{R_{m}}{R_{m}^{\text{max}}}$ corresponds to an achievable rate pair for some beamforming vector. It is worth noting that due to the fact that $R_m$ is strictly concave with respect to $w_m$, the inequality constraints in (\ref{feasibility}) are strictly convex and the test can be performed by standard numerical methods. Once the feasibility is checked, we update the search interval as follows. If the point is feasible, $B_l$ is updated by the current midpoint $\frac{B_l+B_u}{2}$, otherwise, the $B_u$ is replaced with $\frac{B_l+B_u}{2}$. This process is repeated until the difference between the lower bound and the upper bound is within the error tolerance $\delta$, which requires at most $\log_2(\frac{1}{\delta})$ iterations. The pseudocode for the bisection method is shown in Table \ref{Alg:Bisection}.

\begin{algorithm}[t]
\SetAlgoLined
\SetKwInOut{initialization}{Initialization}
\KwData{Disagreement point: $d = \big( R_1^{\text{NE}},\dots, R_M^{\text{NE}} \big)$\newline
Maximum achievable rates: $\big( R_1^{\text{max}},\dots, R_M^{\text{NE}} \big)$\newline
Error threshold: $\delta$}
\KwResult{Optimal rate vector:  $\big( R_1^*,\dots, R_M^* \big)$}
\initialization{$B_l \leftarrow 0$ and $B_u \leftarrow 1$}

\Repeat{$B_u-B_l\leq \delta$}{
$r \leftarrow \frac{B_u+B_l}{2}$\\
$R_1 \leftarrow rR_1^{\text{max}}$\\
Given $R_1$, compute $R_2,\dots,R_M$ by (\ref{KSformula})\\
\eIf{$(R_1, \dots, R_M)$ is feasible according to (\ref{feasibility})}{$B_l \leftarrow r$}{
$B_u \leftarrow r$}
}
$R_1^* \leftarrow rR_1^{\text{max}}$\\
\For{$m\leftarrow 2$ \KwTo $M$}{

$R_m^* \leftarrow R_{m}^{\text{NE}} + \big(R_{m}^{\text{max}}-R_{m}^{\text{NE}}\big)\frac{R_{1}-R_{1}^{\text{NE}}}{R_{1}^{\text{max}}-R_{1}^{\text{NE}}}$
}
\KwRet{ $\big( R_1^*,\dots, R_M^* \big)$}
\caption{Bisection Algorithm for KSBS}
\label{Alg:Bisection}
\end{algorithm}

\section{Simulation Results}\label{sec:simulations}
For simulations, we consider the drone-based communications over $2$ GHz carrier frequency, i.e., $f_c$ = $2$ GHz, in an urban environment with parameters $a= 9.61$, $b = 0.16$~\cite{al2014modeling}. We assume that the minimum allowable received signal power for a successful transmission is $\epsilon = -60$ dBm. We also consider a repository of $12$ DBSs in which there are three different types of DBSs with maximum transmit power  of $35$ dBm, $39$ dBm, and $43$ dBm, and there are four identical DBSs of each kind. The goal is to provide wireless coverage for the ground users that are distributed in a $10$ Km $\times$ $10$ Km area. The simulation parameters are summarized in Table \ref{tab:my-table}.

\begin{table}[t]
\centering
\caption{Simulation parameters.}
\label{tab:my-table}
\begin{tabular}{|c|c|c|c|}
\hline
\textbf{Parameter}               &\textbf{ Value}         & \textbf{Parameter}         & \textbf{Value}          \\ \hline
$a$                     & 9.61          & $f_c$             & 2 GHz          \\ \hline
$b$                     & 0.16          & $\epsilon$        & -60 dBm        \\ \hline
$\eta_\text{LoS}$       & 1 dB          & $L_x$             & 10 Km          \\ \hline
$\eta_\text{NLoS}$      & 20 dB         & $L_y$             & 10 Km          \\ \hline
$K$                     & 10            & $\beta$           & 1 Mbps         \\ \hline
$\sigma_x$              & 20            & N                 & 12             \\ \hline
$\sigma_y$              & 20            & $h^{\text{max}}$  & 3 Km           \\ \hline
$\xi$                   & 0.05          & $P^t_\text{max}$  & $\{35, 39, 43\}$ dB \\ \hline
\end{tabular}
\end{table}

The ground users are distributed randomly in the area. We consider the uniform and truncated Gaussian distributions for users location. Assuming that the x-coordinate and the y-coordinate are independent random variables, these distributions for a rectangular area with size of $L_x\times L_y$ are respectively given by \cite{papoulis2002probability}:
\begin{align}
f^{\text{U}}(x, y) &=\frac{1}{L_{x}L_{y}},\\
f^{\text{tG}}(x, y) &=\frac{1}{G}\exp\left(\frac{L_{x}-\mu_{x}}{\sqrt{2\sigma_{x}}}\right)^{2}\exp\left(\frac{L_{y}-\mu_{y}}{\sqrt{2\sigma_{y}}}\right)^{2},&\tag{21} \end{align}
in which $G = 2\pi \sigma_x \sigma_y \text{erf}\big(\frac{L_x-\mu_x}{\sqrt{2\sigma_x}}\big)\text{erf}\big(\frac{L_y-\mu_y}{\sqrt{2\sigma_y}}\big)$ is the normalizing constant and $\text{erf}(z) = \frac{2}{\sqrt{\pi}}\int_0^z e^{-t^2}dt$ is the Gauss error function. Moreover, $\mu_x$, $\sigma_x$, $\mu_y$, $\sigma_y$ are the mean value and standard deviation in the $x$ and $y$ directions. The truncated Gaussian distribution is used for modeling a hotspot area in which the ground users are concentrated around the hotspot center $(\mu_x, \mu_y)$ and their density decreases as they get further further away from the center \cite{mozaffari2016optimal}.

\begin{figure}
    \centering
    \begin{subfigure}[b]{0.5\textwidth}
        \includegraphics[width=\textwidth]{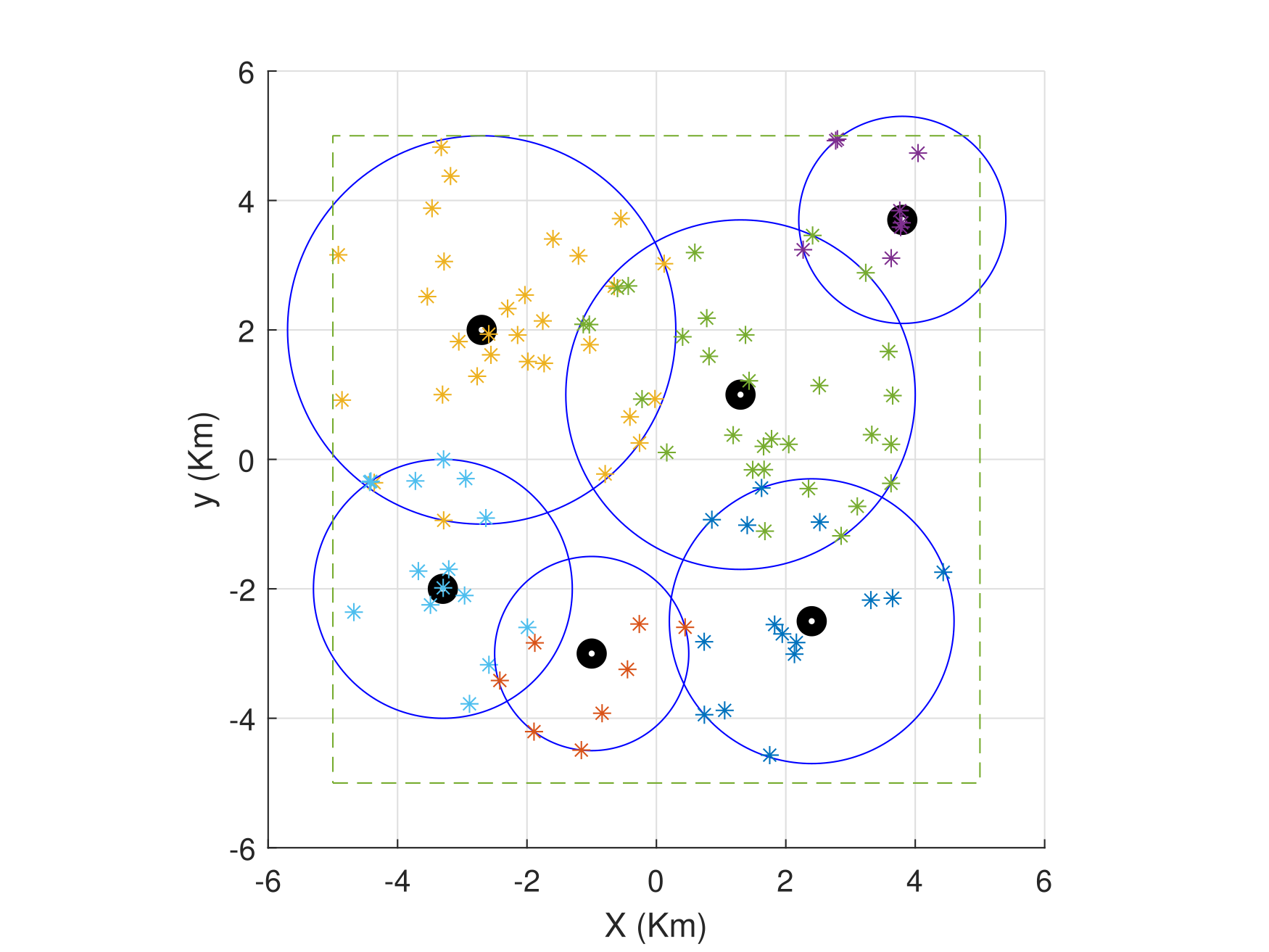}
        \caption{}
        \label{fig:2D}
    \end{subfigure}
    ~ 
    \begin{subfigure}[b]{0.5\textwidth}
        \includegraphics[width=1\textwidth]{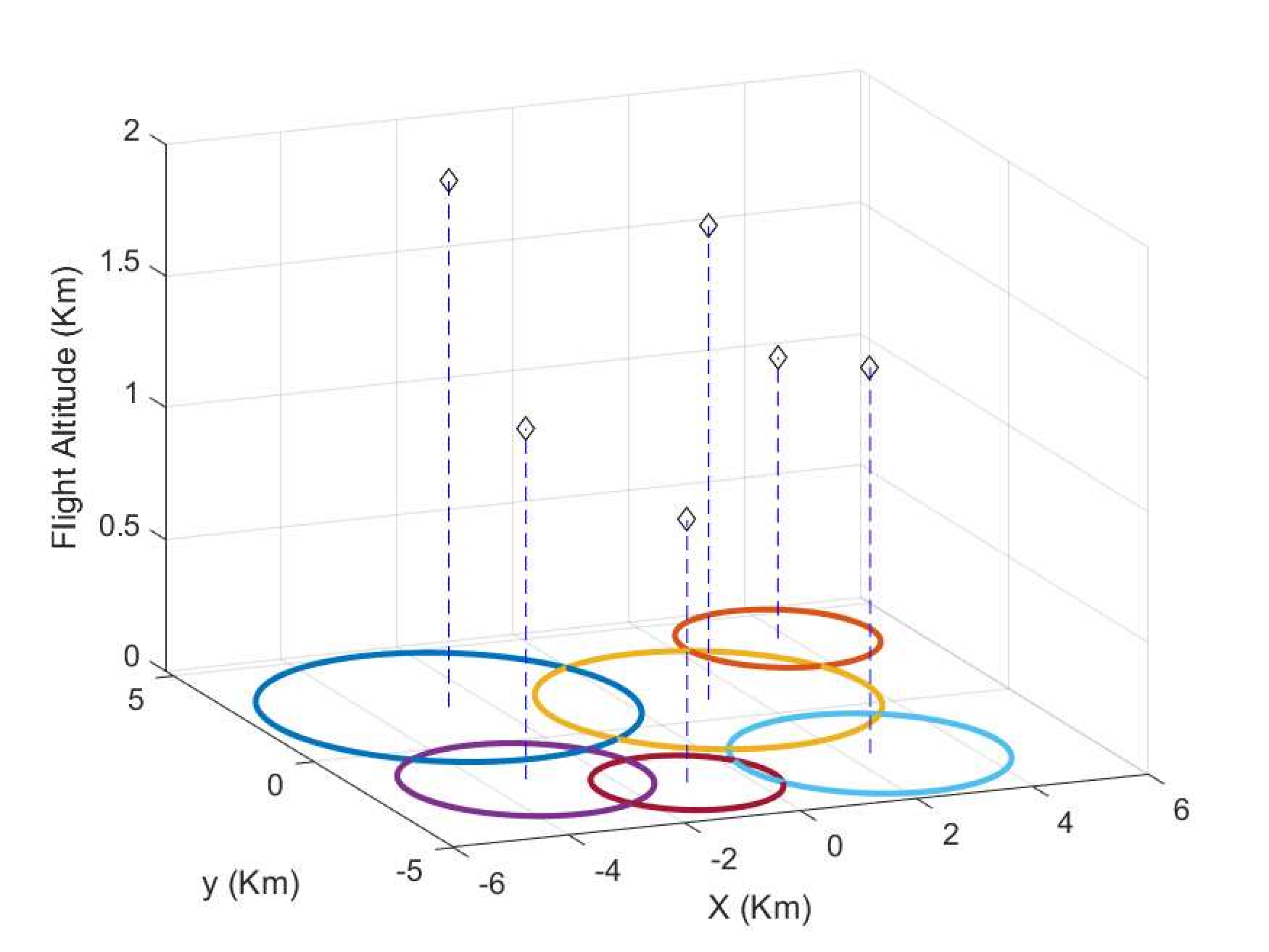}
        \caption{}
        \label{fig:3D}
    \end{subfigure}
    \caption{An illustrative snapshot of the optimal placement of the DBSs and the user-DBS association: (a) User distribution and the 2D projection of the DBSs, and (b) the optimal 3D location of the DBSs and their corresponding coverage disks. The user are illustrated by red dots and uniformly distributed in a 10 km $\times$ 10 km area.}\label{fig:snapshot}
\end{figure}

Fig.~\ref{fig:snapshot} illustrates the optimal resource allocation and 3D placement of the DBSs as well as the user-DBS association for snapshot of the ground users' topology. In particular, Fig.~\ref{fig:2D} shows the 2D projection of the DBSs and their corresponding coverage disks. It can be seen that for each coverage disk, there exists at least two ground users on its boundary. Consequently, one cannot shrink any of these coverage disks without leaving some ground user out of the coverage area. In other words, the DBSs' coverage radii are minimized while providing the required service to the ground users. Fig.~\ref{fig:3D} shows the 3D location of the DBSs. Given the coverage radius of each DBS, its flight altitude has been optimized to minimize the required transmit power. Moreover, it can be seen that only $6$ out of the $12$ DBSs are deployed in this particular illustrative example. Deploying more DBSs will unavoidably decrease the power efficiency by increasing the inter-cell interference.

\begin{figure}
    \centering
    \includegraphics[width=9cm]{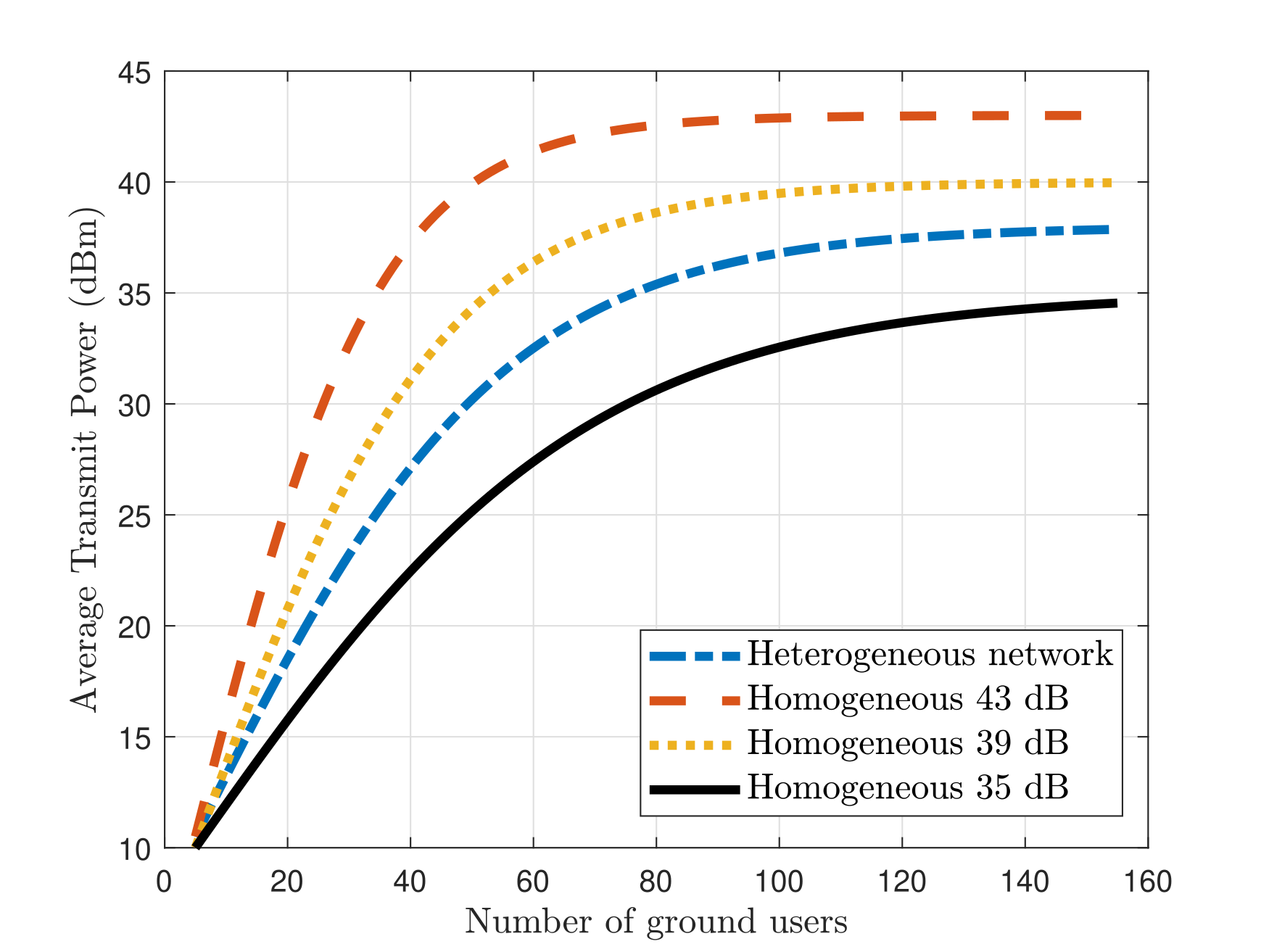}
    \caption{Total average power consumption versus the number of ground users for uniform distribution.}
    \label{fig:powerVSnumber}
\end{figure}

\begin{figure}
    \centering
    \includegraphics[width=9cm]{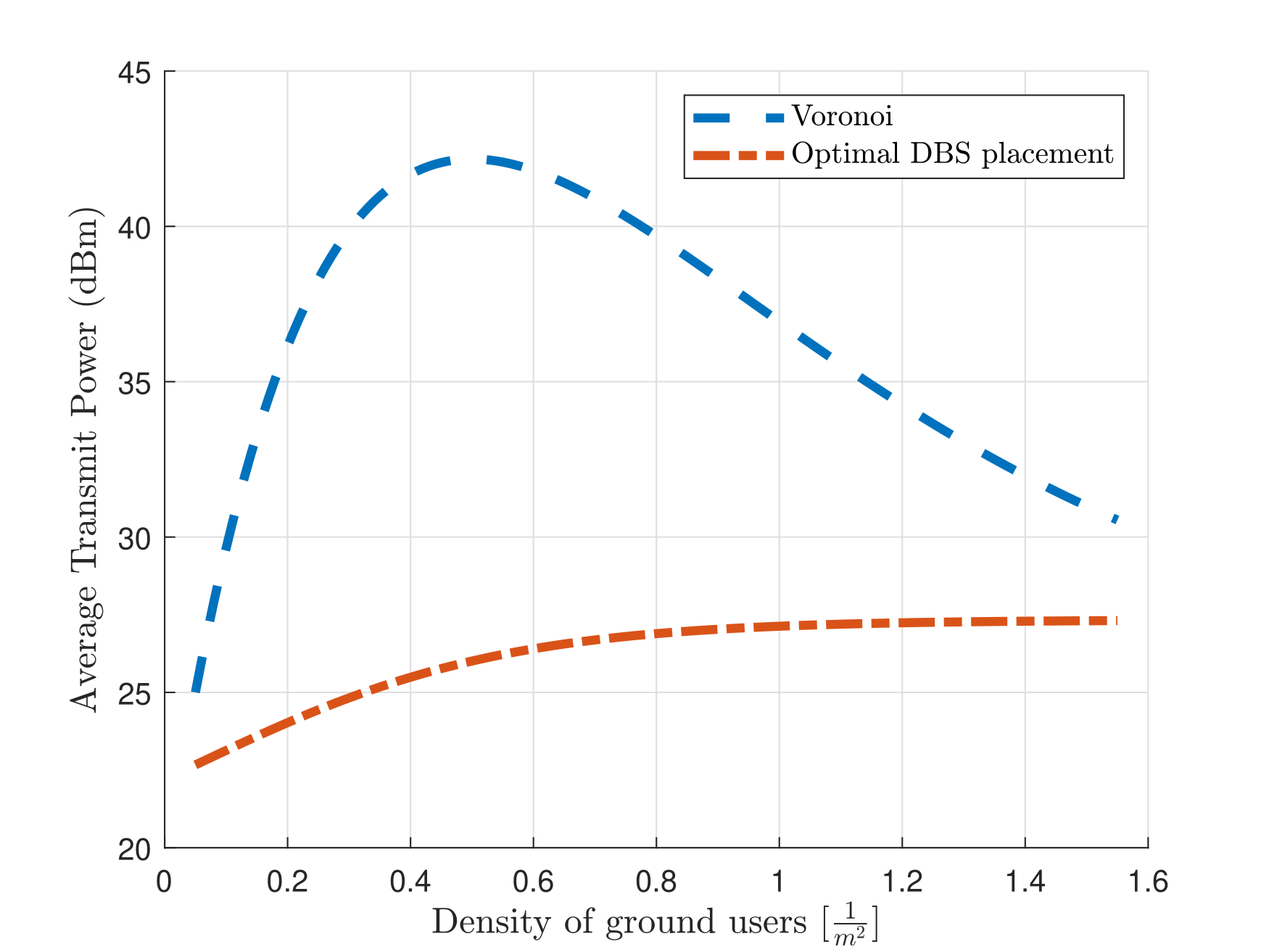}
    \caption{Total average power consumption versus the user density for the truncated Gaussian distribution.}
    \label{fig:powerVSdensity}
\end{figure}

Fig.~\ref{fig:powerVSnumber} shows the average transmit power of the DBSs as a function of the number of ground users which are uniformly distributed in  a $10$ Km $\times$ $10$ Km area. For a better comparison, two scenarios are considered, namely, the heterogeneous network consisting of all $3$ different DBSs and the homogeneous networks consisting of a single type DBS.  In this figure, we can see that as the number of users increases, the average required transmit power of DBSs increases as well. However, by increasing the number of users beyond $100$, the area becomes more and more saturated and thus, the optimal location of the DBSs and their transmit power do not change dramatically which explains the concavity of the transmit power. Moreover, it is seen that as the number of users increases, the DBSs use their maximum transmit power to provide the required QoS for the ground users.

Fig.~\ref{fig:powerVSdensity} shows the average transmit power of the DBSs versus the density of the users for the optimal placement of the DBSs and the Voronoi tessellation. The users are distributed according to the truncated Gaussian distribution with the hotspot center being located at the center of area, i.e., $(\mu_x, \mu_y) = (0,0)$. In order to draw a fair comparison, once the optimal number of DBSs for a given user density is obtained, the area is divided into equal subareas and a DBS is placed at the center of each subarea for the Voronoi cell placement. As we can see, the average transmit power for optimal DBS placement approach is significantly lower than the Voronoi case. According to Fig.~\ref{fig:powerVSdensity}, the Voronoi case is more sensitive to the users' density compared to the optimal DBS placement approach. This is due to the fact that the  DBS placement approach is determined based on the location of the users such that the transmit power is minimized. However, in the Voronoi case, location of the DBSs is set without considering the users' location. As observed in Fig.~\ref{fig:powerVSdensity}, for the low user density case in which the users are more spread over the area, the performance of Voronoi and optimal cell boundaries are close. However, as the density increases, the proposed optimal case becomes better but then they get close again. The reason is that, for highly dense scenarios, the area becomes more saturated and the users spread over a larger area. Thus, the number of DBSs and subareas in the Voronoi tessellation increases and power efficiency for the Voronoi case is improved.

Fig.~\ref{fig:numVSnum} shows the optimal number of DBSs in order to satisfy the coverage requirement of the ground users with minimum average transmit power. In this figure, we can see that number of DBSs is a monotonically increasing function of the number of ground users. However, the number of DBSs does not solely depend on the number of users, it also depends on how the users are distributed in the area. According to Fig.~\ref{fig:numVSnum}, for a large number of users, the required number of DBSs in a congested hotspot scenario is less than the scenario in which the users are evenly distributed in a larger area. However, this result may be misleading at the first glance. Providing coverage with a smaller number of DBSs is not necessarily equivalent to having a better power efficiency. In fact, the optimal number of DBSs and power efficiency are significantly dependent on the topology of the users. Fig.~\ref{fig:numVSnumHomo} shows the the number of DBSs in a homogeneous network as a function of the number of ground users for a uniform distribution. Similar to Fig.~\ref{fig:numVSnum}, we can see that that the number of DBS is a monotonically increasing function of the number of ground users. However, as the power capacity of the DBSs increases, a lower number of DBSs is required to service a given number of users. This is to the the fact that more transmit power corresponds to larger coverage disk on the ground. For instance, according to Fig.~\ref{fig:numVSnumHomo}, for $140$ ground users, the number of required DBSs with transmit power $43$dB, $39$dB, and $35$dB are 5, 6, and 7, respectively.

Fig.~\ref{fig:rate} shows the users' average received data rate versus the number of users for two different distribution models. According to Fig.~\ref{fig:rate}, the average received data rate for a given number of users is significantly lower in hotspot areas. This is in fact due to the severe interference caused by the neighboring DBSs located in the hotspot area compared to the more distant and presumably low-interference DBSs when deployed in a regular area with uniformly distributed users. Although the latter case requires more DBSs according to \ref{fig:numVSnum}, it is more power efficient and offers a higher data rate for the ground users. Fig.~\ref{fig:rateHomo} shows the users' average received data rate in a homogeneous network of DBSs. It shows that as the number of users increases, the average received data rate decreases. In homogeneous networks, the more transmit power the DBSs possess, the higher data rate they can provide for the ground users. However, the higher data rate comes at the price of higher average transmit power (see Fig.~\ref{fig:powerVSnumber}), which is in agreement with our intuition. One interesting feature in Fig.~\ref{fig:rateHomo} is the change of its concavity as the number of users increases. Indeed, by increasing the number of users which demands for more DBSs, the chance of severe intercell interference increases which explains the concavity change in Fig.~\ref{fig:rateHomo}.

\begin{figure}
    \centering
    \begin{subfigure}[b]{0.45\textwidth}
        \includegraphics[width=\textwidth]{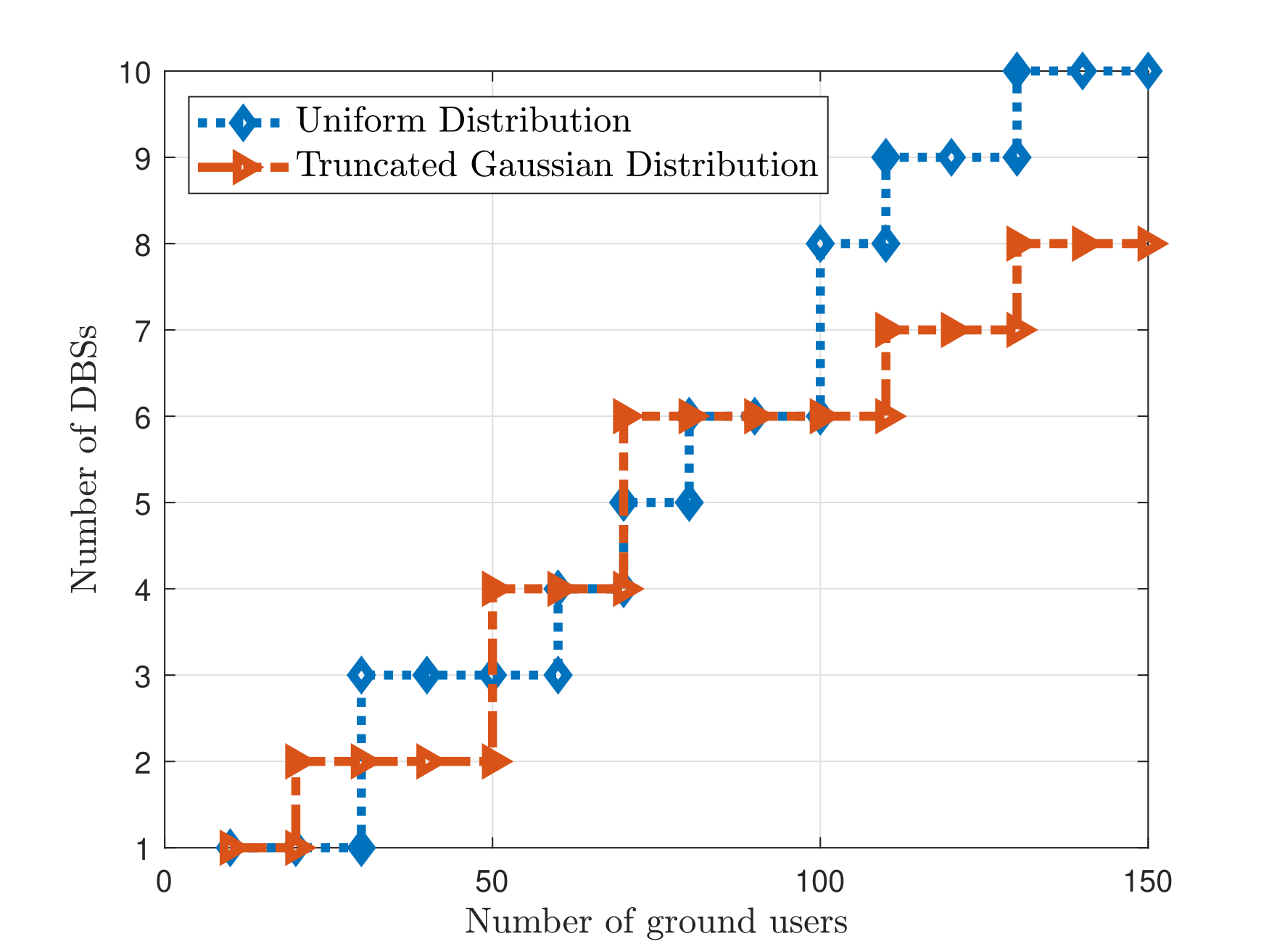}
        \caption{}
        \label{fig:numVSnum}
    \end{subfigure}
    ~ 
    \begin{subfigure}[b]{0.45\textwidth}
        \includegraphics[width=1\textwidth]{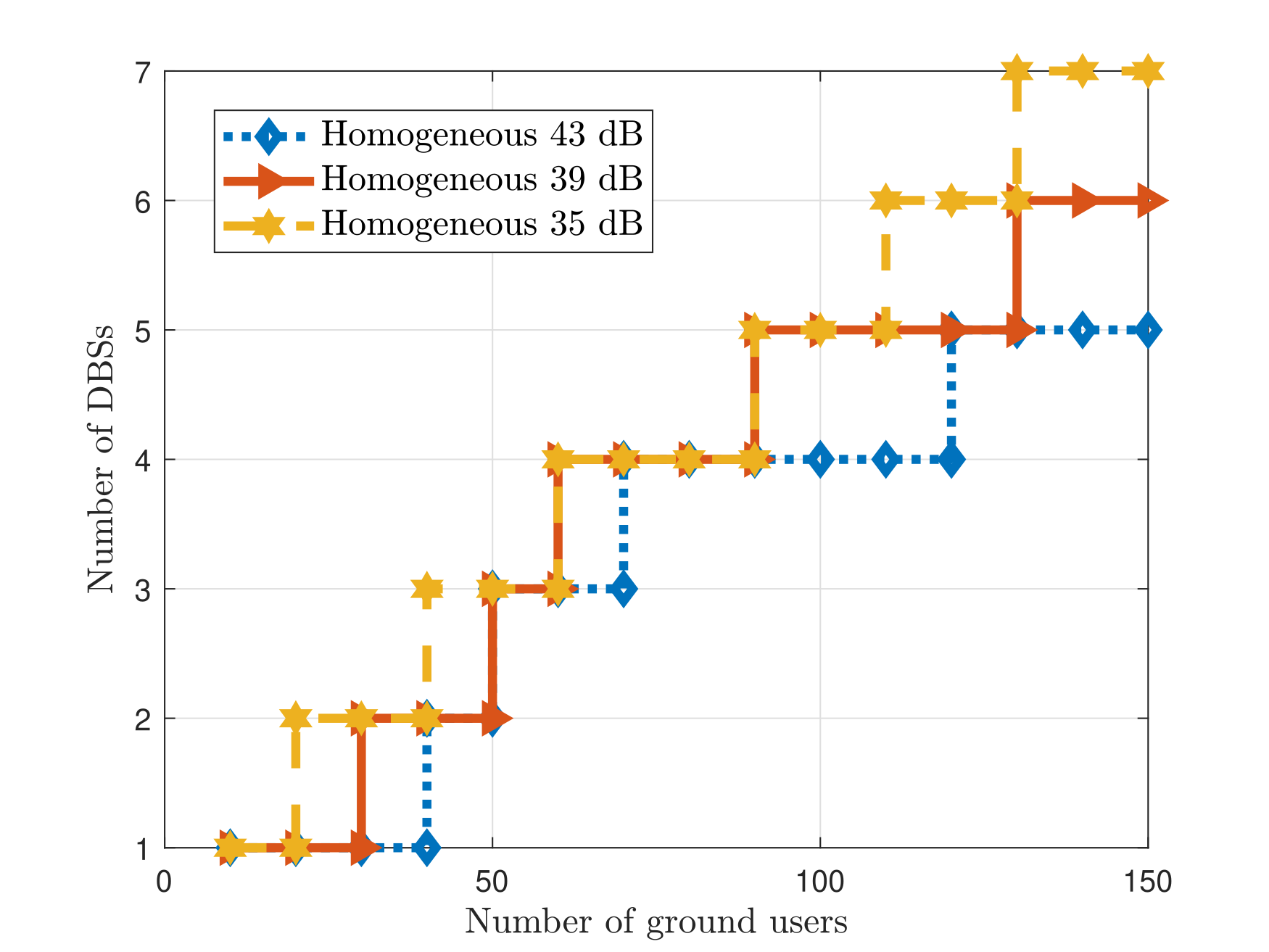}
        \caption{}
        \label{fig:numVSnumHomo}
    \end{subfigure}
    \caption{The number of DBSs vs. the number of ground user: (a) the deployment of the heterogeneous repository of the DBSs to provide service for ground users for two different distributions; and (b) the deployment of identical DBSs (i.e., homogeneous DBS network) for uniform user distribution.}\label{fig:snapshot2}
\end{figure}



\begin{figure}
    \centering
    \begin{subfigure}[b]{0.45\textwidth}
        \includegraphics[width=\textwidth]{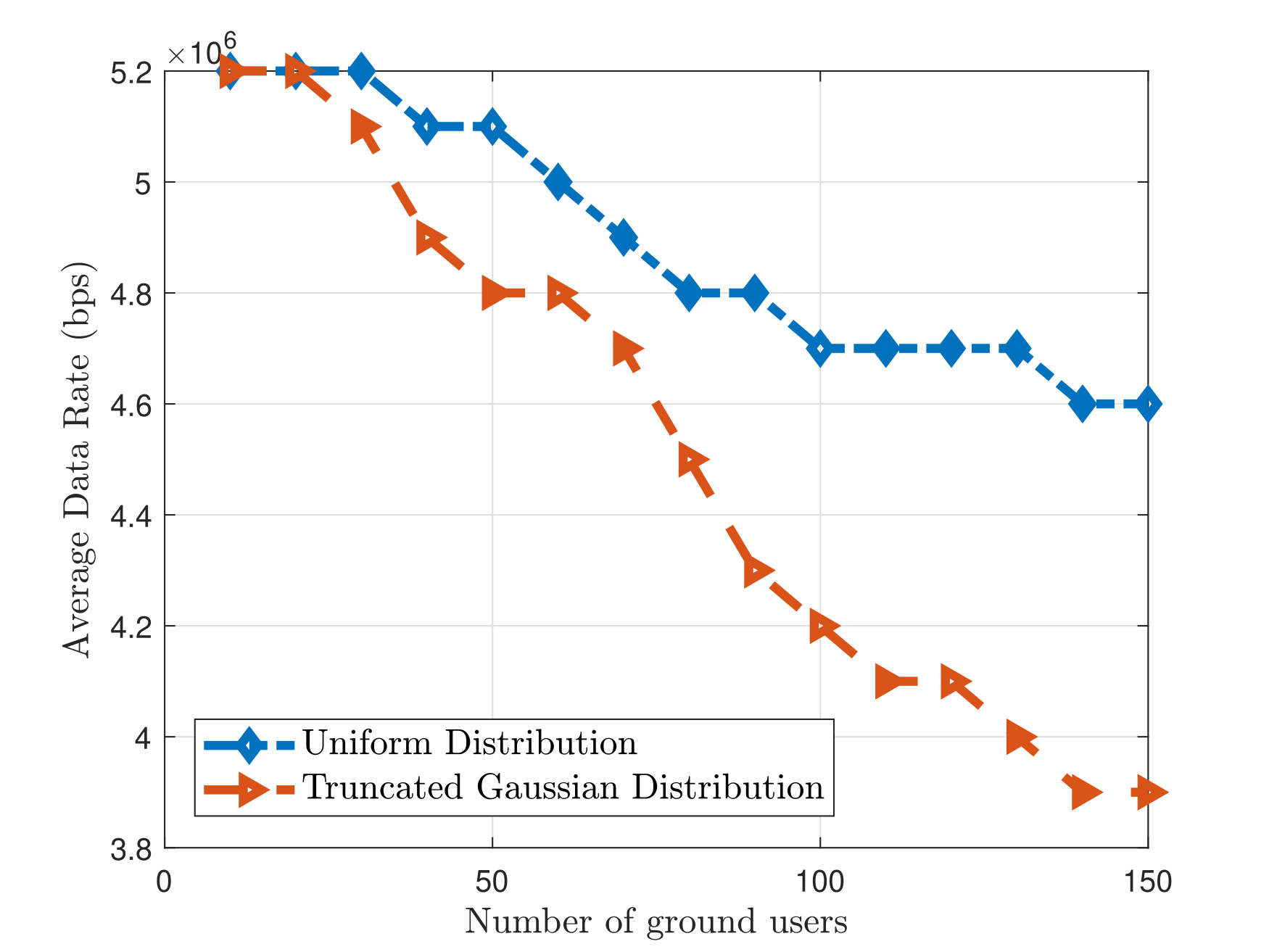}
        \caption{}
        \label{fig:rate}
    \end{subfigure}
    ~ 
    \begin{subfigure}[b]{0.45\textwidth}
        \includegraphics[width=1\textwidth]{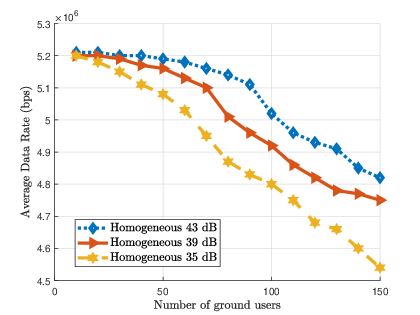}
        \caption{}
        \label{fig:rateHomo}
    \end{subfigure}
    \caption{Average received data rate versus the number of ground users: (a) the deployment of the heterogeneous repository of the DBSs to provide service for ground users for two different distributions; and (b) the deployment of identical DBSs (i.e., homogeneous DBS network) for uniform user distribution.}\label{fig:snapshot3}
\end{figure}



\section{Conclusions}\label{sec:conclusion}
This paper provided a new resource allocation and optimal 3D placement method for drone-based wireless networks. Given a heterogeneous set of DBSs, the proposed method finds the optimal number of DBSs out of the available resources and determines their optimal location to satisfy the ground users' rate requirement while maintaining the minimal aggregate transmit power. The optimization problem to tackle this problem is NP-hard and too complex to solve due to the sheer number of unknowns and the interdependence of optimization variables. Therefore, the optimization problem is decomposed into two sub-problems which are solved iteratively. First, ignoring the intercell interference between the DBSs, we find the optimal subset of DBSs along with their optimal 3D location to cover the ground users based on the received SNR criterion. Next, a Kalai–Smorodinsky bargaining solution is developed to address the intercell interference between the DBSs in a fair manner and increase the users' data rate, which is a function of SINR. These two sub-problems are solved iteratively until the algorithm converges to the optimal solution. It is worth noting that the proposed algorithm requires low to moderate computational resources and offers a practical solution to be implemented on the DBSs. The results have shown the effectiveness of the developed algorithm in satisfying the users' data rate while using minimum transmit energy.

\section*{acknowledgement}
This material is based upon work supported by the Air Force Office of Scientific Research under award number FA9550-20-1-0090 and the National Science Foundation under Grant Numbers CNS-2034218, CNS-2204445, ECCS-2030047 and CNS-1814727.
\bibliography{myref}
\bibliographystyle{IEEEtran}

\end{document}